\documentclass[a4paper,UKenglish,numberwithinsect,thm-restate]{lipics-v2021}
\pdfoutput=1 %

\usepackage{csvsimple} %
\usepackage{mathtools} %
\usepackage{interval} %
\usepackage{stmaryrd} %

\usepackage{tikz} %
\usetikzlibrary{positioning,arrows}

\usepackage[ruled,noline,linesnumbered,longend]{algorithm2e} %

\newcommand{\Ra}{\Rightarrow}
\newcommand{\La}{\Leftarrow}
\newcommand{\Lra}{\Leftrightarrow}

\newcommand{\sra}{\shortrightarrow}

\newcommand{\ud}{\triangleq}
\newcommand{\diff}{\stackrel{\text{\tiny \(\vartriangle\)}}{\Leftrightarrow}}

\newcommand{\sraa}[1]{\rat{#1}}
\newcommand{\exten}[1]{\lceil #1\rceil}

\newcommand{\rat}[1]{\stackrel{\text{\tiny \(#1\)}}{\rightarrow}}

\DeclarePairedDelimiter\tuple{\langle}{\rangle}

\DeclarePairedDelimiter\braces{\{}{\}}

\DeclareMathOperator{\pre}{\mathsf{pre}}
\DeclareMathOperator{\post}{\mathsf{post}}

\DeclareMathOperator{\simu}{{sim}}
\DeclareMathOperator{\Part}{Part}

\DeclareMathOperator{\PreO}{PreO}
\DeclareMathOperator{\Rel}{Rel}

\DeclareMathOperator{\bis}{bis}
\DeclareMathOperator{\Sim}{Sim}

\newcommand{\Pbis}{P_{\bis}}
\newcommand{\Psim}{P_{\simu}}
\newcommand{\Rsim}{R_{\simu}}

\newcommand{\Ubad}{U_{bad}}

\newcommand{\bN}{\mathbb{N}}

\newcommand{\bS}{\mathbb{S}}

\newcommand{\timeout}{\dag}

\newcommand{\subsetarr}{\stackrel{\subset}{\scriptscriptstyle\shortrightarrow}}

\DeclareMathOperator{\RefAlgo}{\mathsf{SymRef}}

\SetKw{Choose}{choose}
\SetKw{suchthat}{such that}
\SetKwBlock{Nif}{nif}{fin}
\SetKw{Skip}{skip}
\SetKw{True}{true}

\newtheorem{problem}[theorem]{Problem}
\newtheorem{assumption}[theorem]{Assumption}

\title{Computing Reachable Simulations} 

\author{Pierre Ganty}{IMDEA Software Institute, Pozuelo de Alarc\'{o}n, Spain}{pierre.ganty@imdea.org}{https://orcid.org/0000-0002-3625-6003}{}

\author{Nicolas Manini}{IMDEA Software Institute, Pozuelo de Alarc\'{o}n, Spain \and Universidad Polit\'{e}cnica de Madrid, Madrid, Spain}{nicolas.manini@imdea.org}{https://orcid.org/0000-0002-7561-3763}{}

\author{Francesco Ranzato}{University of Padova, Padova, Italy}{francesco.ranzato@unipd.it}{https://orcid.org/0000-0003-0159-0068}{}

\authorrunning{P. Ganty, N. Manini and F.Ranzato} %

\Copyright{Pierre Ganty, Nicolas Manini and Francesco Ranzato} 
\ccsdesc[500]{Theory of computation}
\ccsdesc[500]{Theory of computation~Models of computation}
\ccsdesc[500]{Theory of computation~Design and analysis of algorithms}

\keywords{Simulation, Reachability, Reachable Simulation Problem, Algorithm, Symbolic Data Structure.}

\category{} %

\relatedversion{} %

\nolinenumbers %
\hideLIPIcs %

\EventEditors{John Q. Open and Joan R. Access}
\EventNoEds{2}
\EventLongTitle{42nd Conference on Very Important Topics (CVIT 2016)}
\EventShortTitle{CVIT 2016}
\EventAcronym{CVIT}
\EventYear{2016}
\EventDate{December 24--27, 2016}
\EventLocation{Little Whinging, United Kingdom}
\EventLogo{}
\SeriesVolume{42}
\ArticleNo{23}

\begin{document}

\maketitle

\begin{abstract}
We study the problem of computing the reachable principals of simulation preorder and the reachable blocks of simulation equivalence. Following a theoretical investigation of the decidability and complexity aspects of this problem, which in particular highlights a stark contrast with the already settled case of bisimulation, 
we design algorithms to solve this problem by leveraging the idea of interleaving reachability and simulation computation while possibly avoiding the computation of all the reachable states or the whole simulation preorder.
In particular, we put forward a symbolic algorithm processing state partitions and, in turn, relations between their blocks, which is suited for processing infinite state systems.
\end{abstract}

\section{Introduction}\label{sec:alt-intro}

Given a, possibly infinite, labeled transition system $S$,
we study the problem of computing the reachable principals 
of the simulation preorder, i.e.\ the greatest simulation relation on $S$, 
and the reachable blocks of the simulation partition, i.e.\ the  equivalence induced by the simulation preorder on $S$, here called \emph{reachable} (part of) \emph{simulation} (preorder or partition) problem.  
By reachable we mean the principals of the simulation preorder and the 
blocks of the simulation partition that intersect the reachable states of the system, therefore ignoring unreachable principals and blocks which, typically, are of no/negligible interest. 

A na\"{\i}ve solution to our \emph{reachable simulation problem} would be: first, compute the simulation preorder/partition, and 
then, filter out the principals/blocks containing no reachable states.
This, however, would require the computation of the entire simulation preorder and of, possibly, all the reachable states.
In this work, we will present a completely different solution relying on a convoluted interleaving of reachability and simulation computation
that allows us to possibly avoid the computation of all the reachable states
and of the whole simulation preorder.

\subparagraph*{Contributions.} 
In Section~\ref{sec:onlinesim}, we show, by means of a reduction to an undecidable termination problem for infinite state systems, and then for finite state systems through a reduction to the st-connectivity NL-complete problem for directed graphs, that there is stark contrast between the problem of computing the reachable blocks of the bisimulation partition, settled by Lee and Yannakakis~\cite{leeOnlineMinimizationTransition1992} in STOC 1992,  and the analogous problem for the simulation preorder and equivalence tackled by this paper. 
In Section~\ref{sec:soundonline}, we put forward a first algorithm which solves the reachable simulation problem for the simulation preorder while yielding a sound over-approximation for the simulation partition. 
Additionally, we sketch in Section~\ref{subsec:partition} how a variation of this first algorithm allows us to compute the reachable blocks of the simulation partition (due to lack of space, the full description of this variant algorithm is deferred to Appendix~\ref{sec:completeonline}). 
Our results prove correctness and termination of this first algorithm for finite state systems and discuss how correctness is achieved under simple assumptions for the case of infinite state systems as well. 
Moreover, we provide examples showing correctness and termination on some infinite state systems.
Section~\ref{sec:symbolic_algo} introduces the notion of 2PR triple consisting of two state partitions (2P) and of a relation (R) between the blocks of these two partitions, which allows us to design a symbolic algorithm for solving the reachable simulation problem for the simulation preorder, where a relation between states is symbolically represented through a 2PR.
Besides inheriting the termination and correctness guarantees of the first algorithm, we claim that the 2PR-based algorithm terminates faster and more often for the case of infinite state systems.
We support our claims by providing toy examples showing their differences and also by providing empirical evidence where we show (in Appendix~\ref{sec:experiments}) that the symbolic algorithm runs significantly faster than the first explicit algorithm on some benchmarks arising from well-known mutual exclusion protocols, that is, the 2PR-based symbolic approach pays off.

\subparagraph*{Applications.} 
Computing reachable simulation principals and blocks may find several practical applications.  
A noteworthy use case of the reachable principals of the simulation preorder is given by the determinization algorithms $\textsc{Subset}(f)$ and $\textsc{Transset}(f)$ for nondeterministic finite automata designed by van Glabbeek and Ploeger \cite{vanglabbeekFiveDeterminisationAlgorithms2008}. In fact, these procedures can be computationally enhanced by leveraging the simulation preorder  (where $f$ is instantiated to the operator $\mathsf{compress}_{\subsetarr}$ and ${\subsetarr}$ denotes the simulation preorder) restricted to the reachable principals only: %
since the automaton determinization proceeds forward starting from the initial states, the unreachable principals of the simulation preorder are therefore useless. 
It turns out that these simulation-based algorithms compute smaller deterministic automata compared to their plain versions~\cite{vanglabbeekFiveDeterminisationAlgorithms2008}.
In a different application, the reachable blocks of the simulation partition define the states of the reduced system.
The question of computing the transitions between the blocks of the reduced system has been investigated in depth by Bustan and Grumberg~\cite{bustanSimulationBasedMinimization}, who explore the difference and trade-offs of the \(\exists\exists\)  (i.e., $B\sra^{\exists\exists} B'$ if{}f $\exists s\in B\ldotp\exists s'\in B'\ldotp s\sra s'$) and  \(\forall\exists\) definitions (i.e., $B\sra^{\forall\exists} B'$ if{}f $\forall s\in B\ldotp\exists s'\in B'\ldotp s\sra s'$) for transitions between blocks, and, in turn, design an algorithm for computing the minimal number of transitions.

\subparagraph*{Related Work.} 
The closest work to ours is that by Lee and Yannakakis \cite{leeOnlineMinimizationTransition1992}, who first designed in 1992 
an intricate interleaving of reachability and bisimulation computation, here referred to as the LY algorithm.
To put the LY algorithm in its historical context, it was one of several algorithms to compute the reachable part of the bisimulation-based quotiented system~\cite{bouajjaniMinimalModelGeneration1991,bouajjaniMinimalStateGraph1992,leeOnlineMinimizationTransition1992}.
These algorithms have in common the interleaving of the bisimulation computation---which is itself a partition refinement algorithm---with the computation determining which bisimulation blocks are reachable.
The remarkable interest of interleaving reachability and bisimulation computation is that the resulting algorithms terminate as least as often (possibly more often) than the na\"{\i}ve procedure consisting in first computing the bisimulation and next determining its reachable blocks.
Later on, these algorithms for computing the reachable bisimulation have been revisited by Alur and Henzinger in a chapter of their unpublished book on computer-aided verification~\cite[Chapter~4]{alurGraphMinimization}, as well as by the theoretical and experimental comparison made by Fisler and Vardi~\cite{fislerBisimulationMinimizationSymbolic2002}.
Let us also mention that algorithms combining reachability and bisimulation computation inspired by the LY algorithm have been used 
in several different contexts ranging from 
program analysis~\cite{gulavaniSYNERGYNewAlgorithm2006,pasareanuConcreteModelChecking2005} to hybrid systems verification~\cite{majumdarAbstractionbasedControllerDesign2020}.

Our focus on simulation stems from the well-known fact that this behavioural relation provides a better state space reduction than bisimilarity, yet the similarity quotient retains enough precision for checking all linear temporal logic formulas or branching temporal logic formulas without quantifier switches~\cite{BensalemBLS92,hmc18,GrumbergL91,gl94,LoiseauxGSBB95}.
Moreover, infinite state systems like 2D rectangular automata may have infinite bisimilarity quotients, yet their similarity quotients are always finite %
(see \cite{henzingerComputingSimulationsFinite1995} and \cite{henzingerHybridAutomataFinite1995} for the precise details including the initial partition). 
~

There is a large body of work~\cite{BloomP95,Cece17,crt2011,gpp03,gp08,henzingerComputingSimulationsFinite1995,Ranzato2013,ranzato2014,RT07,ranzatoEfficientSimulationAlgorithm2010,TanC01} on efficiently computing the simulation preorder, through both explicit or symbolic algorithms. 
One major interest for simulation algorithms comes from the fact that simulation-based minimization of (labeled) transition systems strongly preserves $\forall$CTL$^*$ formulas \cite{bustanSimulationBasedMinimization,GrumbergL91,gl94}.
Ku\v{c}era and Mayr~\cite{KM02,KuceraM02} have compared simulation and bisimulation equivalence from the perspective of the computational complexity for deciding them, and gave precise justifications to the claim  that 
similarity is computationally harder than bisimilarity.  

To the best of our knowledge, no previous work considered the problem of computing the reachable principals 
of the simulation preorder or the reachable 
blocks of the simulation partition.

Due to lack of space, additional material, proofs included, is deferred to the appendix. 

\section{Background}\label{sec:background}
\subparagraph*{Orders and Partitions.} %
Given a (possibly infinite) set \(\Sigma\), we denote with \(\wp(\Sigma)\) the powerset of \(\Sigma\), and with \(\Rel(\Sigma)\ud \wp(\Sigma\times \Sigma)\) the set of relations over \(\Sigma\).
If \(R\in \Rel(\Sigma)\) then: for \(S\subseteq\Sigma\), \(R(S)\ud \{s'\in \Sigma \mid \exists s\in S\ldotp (s,s')\in R\}\); for \(s\in \Sigma\), the set \(R(s)\ud R(\{s\})\) is the \emph{principal} of \(s\); %
\(\Rel(\Sigma)\ni R^{-1} \ud \{(y, x) \in \Sigma \times \Sigma \mid (x, y) \in R\}\) is the \emph{converse} relation of \(R\).
Moreover, for a given set \(S\subseteq\Sigma\) we denote by \(R^S \ud \{R(x) \in \wp(\Sigma)\mid x \in \Sigma,\,R(x) \cap S \neq \varnothing\}\) the set of principals of $R$ that intersect $S$.
A relation \(R\in \Rel(\Sigma)\) is a \emph{preorder} %
if it is reflexive and transitive, and \(\PreO(\Sigma)\ud \{R\in \Rel(\Sigma) \mid R\ \text{is a preorder}\}\) denotes the set of preorders on \(\Sigma\). 
Moreover, \(R\in \Rel(\Sigma)\) is an \emph{equivalence} on $\Sigma$ if it is a symmetric preorder.
A \emph{partition} of \(\Sigma\) consists of pairwise disjoint nonempty subsets of \(\Sigma\), called \emph{blocks}, whose union is \(\Sigma\), and \(\Part(\Sigma)\) denotes the set of partitions of \(\Sigma\).
We will consider finite partitions (i.e., consisting of finitely many nonempty subsets), unless otherwise specified.
An equivalence relation $R$ on \(\Sigma\) induces a partition of \(\Sigma\) where each block is an equivalence class of $R$, and vice versa. 
Given a partition \(P\in \Part(\Sigma)\), \(P(s)\), \(P(S)\) and \(P^S\)  (for \(S\subseteq\Sigma, s\in\Sigma\)) are well defined thanks to the equivalence induced by \(P\). 
In particular, \(P(s)\) is the block including~\(s\), \(P(S)= \mathop{\cup}\{P(s) \in P \mid s\in S\}\), and \(P^S =  \{P(s) \in P \mid s\in S\}\in \Part(P(S))\). 
Given \(P, Q \in \Part(\Sigma)\), \(P \wedge Q \in \Part(\Sigma)\) denotes the partition obtained by intersecting the underlying equivalence relations (this coincides with the meet in the lattice of partitions w.r.t.\ their ``finer than'' partial order).  
\subparagraph*{Simulation and Bisimulation.}
Let \(G=(\Sigma,I, L,\mathord{\sra})\) be a (labeled) transition system, where \(\Sigma\) is a (possibly infinite yet countable) set of states,
\(I \subseteq \Sigma\) is a subset of \emph{initial states}, 
\(L\) is a finite set of action labels, and 
\(\mathord{\sra} \subseteq \Sigma\times L\times\Sigma\) is the \emph{labeled transition relation}, 
where we denote \((x, a, y) \in \mathord{\sra}\) as \(x \sraa{a} y\), and  
\(x \sra y \diff \exists a \in L \ldotp x \sraa{a} y\). 
When \(L\) is a singleton set, we leverage the previous notation when writing \(x\sra y\), and we also write that \((x, y) \in \mathord{\sra}\) when no ambiguity arises.
Given \(a \in L\), \(\post_a \colon \wp(\Sigma) \sra \wp(\Sigma)\) denotes the usual \emph{successor transformer} \(\post_a(X) \ud \{y \in \Sigma \mid \exists x \in X \ldotp x\sraa{a} y\}\), and, dually, \(\pre_a \colon \wp(\Sigma) \sra \wp(\Sigma)\) is the \emph{predecessor}  \(\pre_a(Y) \ud \{x \in \Sigma \mid \exists y \in Y \ldotp x\sraa{a} y\}\).
Moreover, we define \(\post \colon \wp(\Sigma) \sra \wp(\Sigma)\) as \(\post(X) \ud \cup_{a \in L}\post_a(X)\) and, symmetrically, \(\pre \colon \wp(\Sigma) \sra \wp(\Sigma)\) as \(\pre(X) \ud \cup_{a \in L}\pre_a(X)\).
Thus, \(\post^*(I)\ud\cup_{n\in \bN}\post^n(I)\) is the set of \emph{reachable states} (from the initial states) of \(G\).

Given an (initial) preorder \(R_i\in \PreO(\Sigma)\) (e.g., \(R_i\) can be the equivalence relation induced by the initial states or by some labeling of a Kripke structure), a relation \(R\in \Rel(\Sigma)\) is a \emph{simulation} on \(G\) w.r.t.\ \(R_i\) if: (1) \(R\subseteq R_i\); (2) \((s,t)\in R\) and \(s\sraa{a} s'\) imply \(\exists t'\ldotp\, t\sraa{a} t'\) and \((s',t')\in R\). 
Given two principals \(R(s)\), \(R(s')\) such that \({s \sraa{a} s'}\), \(R(s)\) is \emph{\(a\)-stable} (or simply \emph{stable}) w.r.t.\ \(R(s')\) when \(R(s) \subseteq \pre_a(R(s'))\), otherwise \(R(s)\) is called \emph{\(a\)-unstable} (or simply \emph{unstable}) w.r.t.\ \(R(s')\), and, in this case, \(R(s')\) can \emph{refine} \(R(s)\).
The greatest (w.r.t.\ \(\subseteq\)) simulation relation  on \(G\) exists and turns out to be a preorder called \emph{simulation preorder}  of \(G\) w.r.t.\ \(R_i\), denoted by \(\Rsim\in \PreO(\Sigma)\), while \(\Psim\in \Part(\Sigma)\) is the \emph{simulation partition} 
induced by the \emph{similarity} equivalence \(\Rsim \cap (\Rsim)^{-1}\). 
A relation \(R\in \Rel(\Sigma)\) is a \emph{bisimulation} on \(G\) w.r.t.\ an (initial) partition $P_i\in\Part(\Sigma)$ 
if both \(R\) and \(R^{-1}\) are simulations on \(G\) w.r.t.\ $P_i$.  
The greatest (w.r.t.\ \(\subseteq\)) bisimulation relation  on \(G\) w.r.t.\ $P_i$ exists and turns out to be an equivalence called \emph{bisimulation equivalence} (or bisimilarity), whose corresponding \emph{bisimulation partition} is denoted by  \(\Pbis\in\Part(\Sigma)\).

\section{The Problem of Computing Reachable Simulations}\label{sec:onlinesim}

Firstly, let us  define precisely the problem investigated by this work.  
\begin{problem}[\textbf{Reachable Simulation Problem}]\label{problem1}\ \\ \
\textsc{Given:} A labeled transition system \(G=(\Sigma,I,L,\mathord{\sra})\) and an initial preorder \(R_i\in \PreO(\Sigma)\).\\
\textsc{Compute:} The reachable principals of \(\Rsim\) and the reachable blocks of \(\Psim\), 
where \(\Rsim\) and \(\Psim\) are, resp., the simulation preorder and partition of $G$ w.r.t.\ $R_i$. \lipicsEnd
\end{problem}

One first challenge we face in this problem is related to the notion of reachability.
Let us observe that the notion of reachability for blocks of any partition $P\in\Part(\Sigma)$---such as the partition $\Psim$ 
induced by simulation equivalence---is naturally defined as follows:
\begin{align*}
P^{\post^*(I)} &=\{B\in P \mid B \cap \post^*(I) \neq \varnothing\}\\
&=\{P(s) \in P \mid s \in \post^*(I)\}  \enspace.
\end{align*}
When considering simulations rather than bisimulations, a notion of reachability for principals of a relation is also needed, 
leading to multiple generalizations of reachability for blocks of a partition (viz.\ \(P^{\post^*(I)}\)).
In fact, reachability for principals is not uniquely formulated, and as such, the two following different definitions of \emph{reachable principal} of a reflexive relation \(R \in \Rel(\Sigma)\) can both be considered adequate:
\begin{align}
R^{\post^*(I)} &= \{R(s) \in \wp(\Sigma) \mid s \in \Sigma, \, R(s) \cap \post^*(I) \neq \varnothing\} \enspace,\quad \text{alternatively,} \label{eq:principals_intersect}\\
R_{\text{alt}}^{\post^*(I)} &\ud\{R(s) \in \wp(\Sigma) \mid s \in \post^*(I)\}\enspace . \label{eq:principals_generated}
\end{align} 
Clearly, \(R_{\text{alt}}^{\post^*(I)}\subseteq R^{\post^*(I)}\) holds, and the inclusion may be strict for $\Rsim$ as shown next.

\begin{example}	\label{ex:principals_inclusion_strict}
  Let us consider the transition system \((\Sigma=\{0,1\}, I = \{1\}, L=\{a\}, \mathord{\sra}=\{(1,1)\})\) as depicted below, and the initial preorder \(R_i = \Sigma \times \Sigma\). 
The simulation preorder induced by \(R_i\) is given by \(\Rsim(0) = \{0,1\}\), \(\Rsim(1)=\{1\}\) and the blocks of the corresponding simulation partition \(\Psim\) are represented as boxes in the figure. 

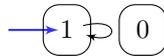
{\centering
\tikz[baseline=(1.south),auto]{
      \tikzstyle{arrow}=[->,>=latex']
      \path      
        (-1,0) node[name=1]{1}
        (0,0) node[name=0]{0}
        ;
      \node (C1) [left of = 1] {};
      \draw[color=blue!80,thick,arrow={-Stealth},shorten >=-3pt, shorten <=3pt] (C1) to (1);
      
      \draw[arrow] (1) edge[loop right] node[right] {} (1);
      
           \path (1.north west) ++(-0.1,0.1) node[name=a1]{} (1.south east) ++(0.1,-0.1) 
      node[name=a2]{};
     \draw[rounded corners=6pt] (a1) rectangle (a2);
     
          \path (0.north west) ++(-0.1,0.1) node[name=b1]{} (0.south east) ++(0.1,-0.1) 
      node[name=b2]{};
     \draw[rounded corners=6pt] (b1) rectangle (b2);
}  
\par}
  
\medskip
\noindent  
Thus, since \(\post^*(I) = \{1\}\), the only reachable principal according to~\eqref{eq:principals_generated} is \(\Rsim(1)\), while the 
reachable principals according to~\eqref{eq:principals_intersect} are both \(\Rsim(0)\) and \(\Rsim(1)\). \lipicsEnd
\end{example}

Let us also remark that the reachable principals, according to \eqref{eq:principals_intersect}, might be infinitely many while, for \eqref{eq:principals_generated}, they might be finitely many.

\subsection{Undecidability for Infinite State Systems}\label{sec:undecidability}
We show that the problem of computing the reachable blocks of a simulation partition $\Psim$ (namely, the subset of blocks $\Psim^{\post^*(I)}=\{B\in \Psim \mid B \cap \post^*(I) \neq \varnothing\}$) over an infinite transition system is, in general, unsolvable even under the assumption that $\Psim$ is a finite partition.
This impossibility result is in stark contrast with the problem of computing reachable blocks for the 
bisimulation partition $\Pbis$, which has been shown solvable by Lee and Yannakakis \cite[Theorem~3.1 and the following paragraph therein]{leeOnlineMinimizationTransition1992}.
To prove our impossibility result, we show that the undecidable halting problem for 2-counter machines can be reduced to the reachability problem for the blocks of a simulation partition, and, in particular, we do so for finite simulation partitions.

A 2-counter machine ($2$-CM, for short) is a tuple \(M = (Q, \Delta, q_0 , H )\), where: \(Q\) is a finite set of states including an initial state \(q_0\), \(H\subseteq Q\) is a set of halting states with \(q_0\notin H\), and \(\Delta\) is a set of transitions. The two counters are \(c_1\) and \(c_2\), and store non-negative integer values (and nothing else).
The set \(\Delta\subseteq Q\times(\{1,2\}\times\{\mathord{+},\mathord{-},\mathord{0}\})\times Q\) comprises three types of transitions:
\begin{alphaenumerate}
  \item \(q \xrightarrow{c_i:=c_i+1} q'\), which increases counter \(c_i\) where \(i\in \{1,2\}\);
  \item \(q \xrightarrow{c_i:=c_i-1} q'\), which decreases counter \(c_i\) where \(i\in \{1,2\}\); if the value of counter $c_i$ is zero then the transition cannot be fired;
  \item \(q \xrightarrow{c_i=0} q'\), which performs a zero-test on counter \(c_i\) where \(i\in \{1,2\}\); if the value of counter $c_i$ is non-zero then the transition cannot be fired.
\end{alphaenumerate}

\noindent
A configuration of \(M\) is a triple \((q, j_1, j_2) \in Q\times \bN^{2}\), where \(q\in Q\) is a state and \(j_1,j_2 \in \bN\) are the values stored by the counters \(c_1\) and \(c_2\) of $M$. 
A configuration \((q, j_1, j_2)\) is halting when \(q\in H\).
The operational semantics of a 2-CM is determined by a transition relation \(\rightarrow\) between configurations and is defined as expected.
The \(2\)-CM \(M\) halts on input \((j_1,j_2)\) if there exist configurations \(\mathit{conf}_1,\ldots,\mathit{conf}_k\) such that: \(\mathit{conf}_1\rightarrow \cdots \rightarrow \mathit{conf}_k\); \(\mathit{conf}_1=(q_0,j_1,j_2)\); \(\mathit{conf}_k\) is halting. 
The halting problem for $2$-CMs is undecidable \cite[Chapter~14]{minsky67}:

\begin{problem}[\textbf{Halting Problem of $2$-CMs}]\ \label{hp2cm} \\ \
\textsc{Given:} A 2-CM \(M\).\\
\textsc{Decide:} Can \(M\) halt on input \((0, 0)\)? \lipicsEnd
\end{problem}

We will show by reduction that since Problem~\ref{hp2cm} is undecidable then so must be the following decision problem:
\begin{problem}\ \label{rse} \\ \
\textsc{Given:} A 2-CM \(M=(Q,C,\Delta,q_0,H)\) and a preorder \(r_i\) on the set $Q$ of states.\\ 
\textsc{Let:} $G=(\Sigma = Q\times \bN^2, I = \{(q_0,0,0)\}, L=\{a\}, \mathord{\rightarrow})$ be the transition system for $M$.\\
\textsc{Let:} $\Psim$ be the simulation partition of \(G\) w.r.t.\
\(R_i \ud \{ \bigl((q,x,y),(q',x',y')\bigr) \mid (q,q')\in r_i\text{ and } x,y,x',y'\in\bN\}\).\\
\textsc{Decide:} \(\forall B\in \Psim \colon \post^*(\{(q_0,0,0)\})\cap B \neq \varnothing\)\,?
\lipicsEnd
\end{problem}

Consider an instance \(M\) of the halting problem for 2-CMs.
We first define a counter machine \(M_1\) that is obtained by adding to \(M\) a new non-halting state \(q_{1}\) for each non-halting state \(q\) of \(M\).
Besides the transitions of \(M\), \(M_1\) has two additional transitions for each new state  \(q_{1}\) added at the previous step: \(q \xrightarrow{c_1:=c_1+1} q_{1}\) and \(q_{1} \xrightarrow{c_1:=c_1-1} q\).
This definition of $M_1$ ensures that every non-halting configuration in \((Q\smallsetminus H)\times\mathbb{N}^2\) can always progress to a non-halting successor configuration in \((Q\smallsetminus H)\times\mathbb{N}^2\).
Observe that adding such states and transitions does not modify whether \(M\) halts: in fact, \(M\) halts if{}f \(M_1\) halts. 
Furthermore, we assume, without loss of generality, that halting states have no outgoing transitions. 
This assumption can be enforced (if needed) because if an halting state \(h\in H\) has outgoing transitions then we define \(M_2\) by duplicating in \(M_1\) the state \(h\) and all its ingoing transitions into a new non-halting state $q_h$, and, then,
we remove all the outgoing transitions out of \(h\).
Again, this transformation does not modify whether \(M\) halts: \(M\) halts if{}f \(M_2\) halts.
We thus consider the transition system induced by $M_2$ with configurations $Q\times \bN^2$ and with a singleton 
set of initial states \(I \ud \{(q_0, 0, 0)\}\).
Next, let us define the preorder \(r_i\) to be \( ( H \times Q ) \cup \bigl( (Q\smallsetminus H) \times (Q\smallsetminus H)\bigr) \in \PreO(Q)\). Hence, it turns out that 
\[
R_i \ud \:\big((H \times \mathbb{N}^2) \times (Q \times \mathbb{N}^2)\big)\; \cup \big(((Q\smallsetminus H) \times \mathbb{N}^2) \times ((Q\smallsetminus H) \times \mathbb{N}^2)\big)\in \PreO(Q\times \bN^2)\enspace .
\]
Next, we show that \(R_i\) is the simulation preorder, i.e.\ $R_i=\Rsim$, because we can prove that for every transition \(\mathit{conf}_x \rightarrow \mathit{conf}_y\) the inclusion \(R_i(\mathit{conf}_x)\subseteq \pre(R_i(\mathit{conf}_y))\) holds. 
Firstly, note that if \(\mathit{conf}_x\) is halting then \(\mathit{conf}_x \rightarrow \mathit{conf}_y\) for no configuration \(\mathit{conf}_y\) (the definition of \(M_2\) guarantees this property). 
If \(\mathit{conf}_x\) is non-halting then we have that \(R_i(\mathit{conf}_x)=(Q\smallsetminus H) \times\mathbb{N}^2\). 
If \(\mathit{conf}_y\) is halting then we have \(R_i(\mathit{conf}_y)= Q\times\mathbb{N}^2\), which, together with the fact that every configuration of \((Q\smallsetminus H)\times\mathbb{N}^2\) has an outgoing transition (this is ensured by \(M_1\)), shows that the inclusion \(R_i(\mathit{conf}_x)\subseteq \pre(R_i(\mathit{conf}_y))\) holds.
The last case to consider is when both \(\mathit{conf}_x\) and \(\mathit{conf}_y\) are non-halting, which implies that \(R_i(\mathit{conf}_x)=R_i(\mathit{conf}_y)=(Q\smallsetminus H)\times\mathbb{N}^2\). 
Here, we find that the inclusion \(R_i(\mathit{conf}_x)\subseteq \pre(R_i(\mathit{conf}_y))\) holds since every element of \((Q\smallsetminus H)\times\mathbb{N}^2\) has an outgoing transition into some non-halting configuration. 

Therefore, \(R_i=\Rsim\) holds, so that the simulation partition \(\Psim\) generated  by the equivalence
 \(\Rsim\cap (\Rsim)^{-1}\) is:
\begin{equation}\label{eq:reduction_psim_finite}
\Psim=\{ H\times\mathbb{N}^2, (Q\smallsetminus H)\times\mathbb{N}^2\}\enspace . 
\end{equation}

Finally, it turns out that \(\{(q_0,0,0)\}\) can reach a halting configuration 
in \(M\) if{}f  it holds \(\post^*(\{(q_0,0,0)\}) \cap (H \times \bN^2) 
\neq \varnothing \land \post^*(\{(q_0,0,0)\})  \cap ((Q\smallsetminus H) \times \bN^2)\neq \varnothing\), thus if{}f \(\forall B\in \Psim \colon \post^*(\{(q_0,0,0)\})\cap B \neq \varnothing\). 
We have therefore shown the following result. 
\begin{theorem}\label{thm:undecidability}
Problem~\ref{rse} is undecidable.
\end{theorem}

As a consequence, there exists no algorithm that, for a transition system $G$ and initial 
preorder $R_i$, is able to compute the reachable blocks of the simulation partition $\Psim$ of $G$ w.r.t.~$R_i$.
In particular, the proposed reduction considers a case where the simulation partition (that is, \eqref{eq:reduction_psim_finite}) is finite, meaning that even under the hypothesis that 
$\Psim$ consists of a finite set of blocks, the problem still remains undecidable.
This negative result is to be contrasted with the positive result of Lee and Yannakakis for the bisimulation partition $\Pbis$
stating that the LY algorithm terminates when \(\Pbis\) is finite \cite[Theorem~3.1 and the following paragraph therein]{leeOnlineMinimizationTransition1992}. 
Intuitively, the above reduction does not work for the corresponding problem replacing simulation with bisimulation because, in general, we cannot define a finite bisimulation that split halting and non-halting states as we did for \eqref{eq:reduction_psim_finite}.

\subsection{Complexity for Finite State Systems}\label{sec:complexity_finite}
We show a further key difference between the problems of deciding reachability of blocks of \(\Psim\) and \(\Pbis\) over \emph{finite} transition systems. These two problems are obviously both decidable, since we can simply compute independently \(\post^*(I)\) and $\Pbis$/$\Psim$, and, then check whether \(\post^*(I) \cap B=\varnothing\) holds for every block \(B\) in $\Pbis$ or $\Psim$.

For the case of bisimulation, deciding the reachability of $B\in \Pbis$ is solved in \(O(|\Pbis^{\post^*(I)}|)\) time by leveraging the definition of bisimulation: %
picking a pair of bisimilar states $x$ and $y$, i.e.\ such that \(\Pbis(x) = \Pbis(y)\),  it turns out that \(\{B\in \Pbis \mid \post(x)\cap B\neq \varnothing\}= \{B\in \Pbis \mid \post(y)\cap B\neq \varnothing\}\), namely, for $B\in \Pbis$, $x$ can reach $B$ in one transition if{}f $y$ can reach $B$ in one transition, hence picking one state per block of $\Pbis$ suffices. 

For the simulation case, deciding the reachability of $B\in\Psim$ is more involved than in the bisimulation case. 
As the following example suggests, to decide whether $B\in\Psim$ is reachable it seems unavoidable to have to decide whether there is a path of arbitrary length in the transition system reaching~\(B\).

\begin{example}\label{ex:unavoidable}
Let $G_1$, $G_2$ be the following transition systems (resp., left and right in the picture), 
and consider $G_1$  where $I=\{1\}$ and $L$ is a singleton (so the label is omitted).
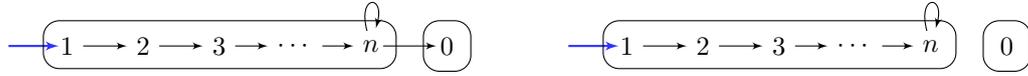
\tikz[baseline=(1.south),auto]{
      \tikzstyle{arrow}=[->,>=latex']
      \path       
      (0,0) node[name=1]{1}
      (1,0) node[name=2]{2}
      (2,0) node[name=3]{3}
      (3,0) node[name=4]{$\cdots$}
      (4,0) node[name=7]{$n$}
      (5,0) node[name=00]{0}
      ;
      \node (C1) [left of = 1] {};
      \draw[color=blue!80,thick,arrow={-Stealth},shorten >=-3pt, shorten <=3pt] (C1) to (1);
      \draw[arrow,shorten >=0pt, shorten <=0pt] (1) to (2);
      \draw[arrow,shorten >=0pt, shorten <=0pt] (2) to (3);
      \draw[arrow,shorten >=0pt, shorten <=0pt] (3) to (4);
      \draw[arrow,shorten >=0pt, shorten <=0pt] (4) to (7);
      \draw[arrow] (7) edge[loop above] node[right] {} (7);
      \draw[arrow,shorten >=-2pt, shorten <=-2pt] (7) to (00);
                        
     \path (1.north west) ++(-0.1,0.1) node[name=a1]{} (7.south east) ++(0.1,-0.1) 
      node[name=a2]{};
     \draw[rounded corners=6pt] (a1) rectangle (a2);
     
     \path (00.north west) ++(-0.1,0.1) node[name=z1]{} (00.south east) ++(0.1,-0.1) 
      node[name=z2]{};
     \draw[rounded corners=6pt] (z1) rectangle (z2);
}
  \qquad   
\tikz[baseline=(1.south),auto]{
      \tikzstyle{arrow}=[->,>=latex']
      \path     
      (0,0) node[name=1]{1}
      (1,0) node[name=2]{2}
      (2,0) node[name=3]{3}
      (3,0) node[name=4]{$\cdots$}
      (4,0) node[name=7]{$n$}
      (5,0) node[name=00]{0}
      ;
      \node (C1) [left of = 1] {};
      \draw[color=blue!80,thick,arrow={-Stealth},shorten >=-3pt, shorten <=3pt] (C1) to (1);
      \draw[arrow,shorten >=0pt, shorten <=0pt] (1) to (2);
      \draw[arrow,shorten >=0pt, shorten <=0pt] (2) to (3);
      \draw[arrow,shorten >=0pt, shorten <=0pt] (3) to (4);
      \draw[arrow,shorten >=0pt, shorten <=0pt] (4) to (7);
      \draw[arrow] (7) edge[loop above] node[right] {} (7);

     \path (1.north west) ++(-0.1,0.1) node[name=a1]{} (7.south east) ++(0.1,-0.1) 
      node[name=a2]{};
     \draw[rounded corners=6pt] (a1) rectangle (a2);
     
     \path (00.north west) ++(-0.1,0.1) node[name=z1]{} (00.south east) ++(0.1,-0.1) 
      node[name=z2]{};
     \draw[rounded corners=6pt] (z1) rectangle (z2);
}

\medskip
\noindent
Here, we have that $\Psim$ w.r.t.\ $R_i=\Sigma\times\Sigma$ 
is given by the two boxes depicted in the left diagram, because $\Rsim(0)=\interval{0}{n}$ and, for all $k\in \interval{1}{n}$, $\Rsim(k)=\interval{1}{n}$, so that $\Psim=\{\interval{0}{0},\interval{1}{n}\}$ and, in turn, the block $\interval{0}{0}\in \Psim$ is actually reachable.

\noindent
Then, define \(G_2\) by removing the transition $n\sra 0$ from $G_1$.
In this case, $\Rsim$ and, therefore, $\Psim$  remain the same as above, however the block $\interval{0}{0}\in\Psim$ becomes unreachable.  

\noindent
Hence, to distinguish whether the block \(\interval{0}{0}\) is reachable or not in these two instances $G_1$ and $G_2$, we have to detect that in \(G_1\) the state \(0\) is reachable, while in \(G_2\) it is not.

\noindent
Let us remark that this is not the case for bisimulation, because we have that $\Pbis=\{\interval{0}{0},\interval{1}{1},\interval{2}{2},\ldots,\interval{n}{n}\}$ for $G_1$, 
while $\Pbis=\Psim$ for $G_2$.  \lipicsEnd
\end{example}

Let us now turn Example~\ref{ex:unavoidable} into a formal complexity argument by showing that, for finite transition systems with a two block (bi)simulation partition (this same argument of course applies to any given finite number of blocks), deciding whether both blocks are reachable is NL-hard for the simulation partition while it is in L for the bisimulation partition. This problem is formally stated as follows.

\begin{problem}\ \label{rsdpfs} \\ \
\textsc{Given:} A finite transition system \(G=(\Sigma,I,L,\mathord{\sra})\) and a two block simulation partition \(\Psim=\{B_0,B_1\}\) of \(G\).\\%
\textsc{Decide:} \(\post^*(I)\cap B_0 \neq \varnothing\) and \(\post^*(I)\cap B_1 \neq \varnothing\)\,? \lipicsEnd
\end{problem}

\begin{restatable}{theorem}{nlhardnesslmembership}
\label{thm:nlhardness}
Problem~\ref{rsdpfs} is {\rm{NL}}-hard, while the analogous decision problem 
for bisimulation partition is in {\rm{L}}. 
\end{restatable}

The NL-hardness proof reduces from the st-connectivity problem in leveled directed graphs, while the L membership proof crucially leverages the definition of bisimulation.
Intuitively, the situation described in Example~\ref{ex:unavoidable} cannot occur anymore for the case of bisimulation because its definition mandates that if a state of the block \(\interval{1}{n}\) has a transition to a state in \([0]\) then each state of \(\interval{1}{n}\) has to have a transition to \([0]\).
Roughly speaking, it turns out that if there is a path of arbitrary length then there is path of length one.
A detailed proof of Theorem~\ref{thm:nlhardness} is given in Appendix~\ref{sec:app-complexity}.

\begin{algorithm}[tbh]
\caption{\textsf{Preorder-based Algorithm}}\label{algo:explicit_sound}
 { \color{green!40!black}
\KwIn{A transition system \(G=(\Sigma,I,L,\mathord{\sra})\), an initial preorder \(R_i\in \PreO(\Sigma)\), an initial finite set \(\sigma_i \subseteq \post^*(I)\).}
}
\(\Rel(\Sigma) \ni R := R_i\)\;
\(\wp(\Sigma) \ni \sigma:=\sigma_i\)\;
\While(){\True}{ 
  { \color{blue!75!black} {\tcp{{\rm \textsc{Inv\textsubscript{1}}}: \(\forall x\in \Sigma\ldotp \Rsim(x) \subseteq R(x)\subseteq R_i(x)\)}}
  \tcp{{\rm \textsc{Inv\textsubscript{2}}}: \(\sigma_i \subseteq \sigma \subseteq \post^*(I)\)}
  \tcp{{\rm \textsc{Inv\textsubscript{3}}}: \(\forall x\in\Sigma\ldotp x\in R(x)\)}
  }
\(U:= \{R(x) \mid R(x) \cap \sigma = \varnothing,\, R(x) \cap (I\cup \post(\sigma))\neq \varnothing\)\};

\(V:= \{\tuple{a,x,x'} \in L \times \Sigma^2 \mid R(x) \cap \sigma \neq \varnothing, x\sraa{a} x',\, R(x) \not\subseteq  \pre_a(R(x'))\}\)\;

\Nif{\label{loc:ex_nif_begin}
	\((U \neq \varnothing) \longrightarrow \mathit{Search}:\) \label{loc:ex_search_begin}\\ 
	\Indp
		\Choose \(R(x)\in U,\, s \in  R(x)\cap (I \cup \post(\sigma))\)\;
		\(\sigma:= \sigma \cup \{s\}\)\;\label{loc:ex_sigmaupdate} \label{loc:ex_search_end}
	\Indm

  \((V \neq \varnothing) \longrightarrow \mathit{Refine}:\) \label{loc:ex_refine_begin}\\   
	\Indp
		\Choose \(\tuple{a, x,x'}\in V\)\;
		\( R(x) := R(x) \cap \pre_a(R(x'))\)\;\label{loc:ex_refine_principal} \label{loc:ex_refine_end}
	\Indm
	
  \((U = \varnothing \land V = \varnothing) \longrightarrow {}\)\Return{\(\tuple{R, \sigma}\)}\label{loc:ex_return}\;
}\label{loc:ex_nif_end}

}
\end{algorithm}

\section{A Reachable Simulation Algorithm}\label{sec:soundonline}
We put forward Algorithm~\ref{algo:explicit_sound} which, given a transition system \(G\), an initial preorder \(R_i\) and an initial set of reachable nodes \(\sigma_i\), where \(\sigma_i\) can be empty, computes 
the reachable principals of \(\Rsim\) according to definition~\eqref{eq:principals_intersect}, along with a sound overapproximation of the reachable blocks of \(\Psim\). 

This algorithm maintains a current relation $R\in\Rel(\Sigma)$ specified through its principals $R(x)\in \wp(\Sigma)$, and a set $\sigma\subseteq\Sigma$ containing states reachable from $I$, so that $R^\sigma = \{R(x) \mid R(x)\cap \sigma \neq \varnothing\}$ includes the principals of $R$ which are currently known to be reachable, i.e., containing a reachable state. 
The algorithm computes the set $U$ of principals that can be added to $R^\sigma$ and the set $V$ representing unstable pairs of principals.
A principal $R(x)$ belongs to $U$ if it contains an initial state or a successor of a state already known to be reachable.
A triple $\tuple{a,x,x'}$ belongs to $V$ if $R(x)$ is known  to be reachable and it can be refined by $R(x')$, i.e., $x\sraa{a} x'$ and $R(x) \not\subseteq \pre_a(R(x'))$. 
Algorithm~\ref{algo:explicit_sound} is presented  in \emph{logical form}, meaning that in this pseudocode we do not require or provide a specific representation for the transition system $G$ or for the sets maintained by the algorithm, namely the state relation $R$, the set of reachable states $\sigma$, and the sets $U$ and $V$. 
The questions around specific representations will be addressed in Section~\ref{sec:2pr}.

Algorithm~\ref{algo:explicit_sound} either updates the reachability information for some principal in
$U$ or stabilizes the pair of principals associated to $\tuple{a,x,x'}\in V$ by refining $R(x)$. 
The pseudocode of Algorithm~\ref{algo:explicit_sound}  uses a \(\KwSty{nif}\) statement %
which is a nondeterministic choice between guarded commands.
In that \(\KwSty{nif}\) statement we have three guarded commands: 
either the \emph{Search} (lines~\ref{loc:ex_search_begin}--\ref{loc:ex_search_end}) or the \emph{Refine} procedures (lines~\ref{loc:ex_refine_begin}--\ref{loc:ex_refine_end}) are executed, or, at line~\ref{loc:ex_return}, when 
both their guards are false, 
the return statement is taken. 
Thus, every execution consists of an arbitrary interleaving of \emph{Search} and \emph{Refine}, possibly followed by the return at line~\ref{loc:ex_return}.
Observe that the guards are such that when the algorithm terminates neither \emph{Search} nor \emph{Refine} are enabled. 

A principal $R(x)$ is refined at line~\ref{loc:ex_refine_principal} only if it is known to be currently reachable. 
Upon termination, $R^\sigma$ turns out to be precisely the set of reachable principals of $\Rsim$ (cf.~\eqref{corr-stm-expl_sound} below).
However, $R$ may well contain unstable principals, so that, in general, $R$  and $\Rsim$ do not coincide.
Turning to the simulation partition, Algorithm~\ref{algo:explicit_sound} yields a sound overapproximation of \(\Psim^{\post^*(I)}\) (cf.~\eqref{corr-stm-expl_sound2} below).

\begin{restatable}[\textbf{Correctness of Algorithm~\ref{algo:explicit_sound}}]{theorem}{correctnessExplicit}
\label{thm:explicit_sound}
Let $\tuple{R,\sigma}\in \Rel(\Sigma)\times\wp(\Sigma)$ be the output of Algorithm~\ref{algo:explicit_sound} on input \(G\) with \(|\Sigma|\in\bN\), $R_i\in \PreO(\Sigma)$ and \(\sigma_i \subseteq \post^*(I)\).
  Let $\Rsim\in \PreO(\Sigma)$ and $\Psim\in \Part(\Sigma)$ be, resp., the simulation preorder and partition w.r.t.\ $R_i$. 
  Let \(P\ud \{y \in \Sigma \mid R(y) = R(x)\}_{x\in \Sigma}\in \Part(\Sigma)\).
  Then:
  \begin{align}
    \Rsim^{\post^*(I)} &= R^\sigma,\tag{1.a}\label{corr-stm-expl_sound}\\
    \Psim^{\post^*(I)}
                       &\subseteq
                       \{B \in P \mid R(B)\cap \sigma \neq \varnothing\}\enspace . \tag{1.b}\label{corr-stm-expl_sound2}
  \end{align} 
\end{restatable}

The reader might find surprising that \eqref{corr-stm-expl_sound2} is not \(\Psim^{\post^*(I)} \subseteq P^\sigma\).
Indeed, \(\Psim^{\post^*(I)} \subseteq P^\sigma\) does not hold in general as shown by the following example:

\begin{example}\label{ex:Psigmanotenough}
Let us consider the system $(\Sigma=\{1,2\},I=\{1\},L,\mathord{\sra}=\{(1,2)\})$, where \(L\) is a singleton, as depicted below.

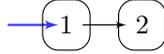
{\centering
\tikz[baseline=(1.south),auto]{
      \tikzstyle{arrow}=[->,>=latex']
      \path      
        (-1,0) node[name=1]{1}
        (0,0) node[name=2]{2}
        ;
      \node (C1) [left of = 1] {};
      \draw[color=blue!80,thick,arrow={-Stealth},shorten >=-3pt, shorten <=3pt] (C1) to (1);
      
      \draw[arrow,shorten >=0pt,shorten <=0pt,] (1) to (2);
      
           \path (1.north west) ++(-0.1,0.1) node[name=a1]{} (1.south east) ++(0.1,-0.1) 
      node[name=a2]{};
     \draw[rounded corners=6pt] (a1) rectangle (a2);
     
          \path (0.north west) ++(-0.1,0.1) node[name=b1]{} (0.south east) ++(0.1,-0.1) 
      node[name=b2]{};
     \draw[rounded corners=6pt] (b1) rectangle (b2);
} 
\par}

\medskip
\noindent
Moreover, fix \(\sigma_i = \varnothing\) and $R_i=\{(1,1), (2,1), (2,2)\} = \Rsim$. 
Here, we have that $\post^*(I)=\{1, 2\}$, and $\Psim=\{\braces{1},\braces{2}\}$, so that the blocks of \(\Psim\) are depicted as boxes in the picture. 
Algorithm~\ref{algo:explicit_sound} on input $R_i$ and $\sigma_i$ returns  $R=\Rsim$, $\sigma=\{1\}$. 
It follows that \(\Psim^{\post^*(I)} = \Psim\), \(P^\sigma = \{\braces{1}\}\), so that \(\Psim^{\post^*(I)} \not\subseteq P^\sigma\) holds.
Since \(\{B \in P \mid R(B)\cap \sigma \neq \varnothing\} = \Psim\), we have that \eqref{corr-stm-expl_sound2} holds, as guaranteed by Theorem~\ref{thm:explicit_sound}. \lipicsEnd
\end{example}

Moreover, the following example shows that the containment of \eqref{corr-stm-expl_sound2} may be strict.
\begin{example}\label{ex:shortcoming}
Let us consider the system $(\Sigma=\{1,2\},I=\{1\},L,\mathord{\sra}=\{(1,1)\})$, where \(L\) is a singleton, as depicted below.

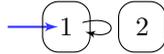
{\centering
\tikz[baseline=(1.south),auto]{
      \tikzstyle{arrow}=[->,>=latex']
      \path      
        (-1,0) node[name=1]{1}
        (0,0) node[name=2]{2}
        ;
      \node (C1) [left of = 1] {};
      \draw[color=blue!80,thick,arrow={-Stealth},shorten >=-3pt, shorten <=3pt] (C1) to (1);
      
      \draw[arrow] (1) edge[loop right] node[right] {} (1);
      
           \path (1.north west) ++(-0.1,0.1) node[name=a1]{} (1.south east) ++(0.1,-0.1) 
      node[name=a2]{};
     \draw[rounded corners=6pt] (a1) rectangle (a2);
     
          \path (0.north west) ++(-0.1,0.1) node[name=b1]{} (0.south east) ++(0.1,-0.1) 
      node[name=b2]{};
     \draw[rounded corners=6pt] (b1) rectangle (b2);

} 
\par}

\medskip
\noindent
Also, let \(\sigma_i = \varnothing\) and $R_i=\{1,2\}\times \{1,2\}$. 
Here, it turns out that $\post^*(I)=\{1\}$, $\Rsim(1)=\{1\}$, $\Rsim(2)=\{1,2\}$, so that $\Psim=\{\{1\},\{2\}\}$, whose blocks are depicted as boxes in the picture. 
Algorithm~\ref{algo:explicit_sound} on input $R_i$ and $\sigma_i$ outputs  $R=\Rsim$, $\sigma=\{1\}$. 
Thus, the containment \eqref{corr-stm-expl_sound2} turns out to be strict since $\{\{1\}\}\subsetneq \{\{1\}, \{2\}\}$.\lipicsEnd
\end{example}

It turns out that Algorithm~\ref{algo:explicit_sound} always terminates on finite state systems.
\begin{restatable}[\textbf{Termination of Algorithm~\ref{algo:explicit_sound}}]{theorem}{terminationExplicit}
\label{thm:algo_sound_finite_termination}
  Let \(G = (\Sigma, I,L, \sra)\) with \(|\Sigma|\in\bN\), \(R_i\in \PreO(\Sigma)\) and \(\sigma_i \subseteq \post^*(I)\). Then, Algorithm~\ref{algo:explicit_sound} terminates on input \(G\), \(R_i\), and \(\sigma_i\).
\end{restatable}

Algorithm~\ref{algo:explicit_sound} may also correctly terminate on infinite state systems, as shown in the next example.
In Appendix~\ref{sec:appendix_extension} we discuss how the correctness proof of Theorem~\ref{thm:explicit_sound} can be extended 
to infinite state systems.
\begin{example}\label{ex:algo_sound_termination}
  Consider the following transition system with infinitely many states where \(I=\{0\}\) is the initial state, \(L\) is a singleton, and the boxes gather states sharing the same principal in \(R_i\).  

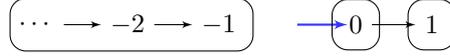
{\centering
\tikz[baseline=(1.south),auto]{
      \tikzstyle{arrow}=[->,>=latex']
      \path      
        (1,0) node[name=0]{0}
        (2,0) node[name=1]{1}
        (-3.2,0) node[name=dots]{$\cdots$}
        (-2,0) node[name=m2]{$-2$}
        (-0.8,0) node[name=m1]{$-1$}
        ;
      \node (C1) [left of = 0] {};
      \draw[color=blue!80,thick,arrow={-Stealth},shorten >=-3pt, shorten <=3pt] (C1) to (0);
      \draw[arrow,shorten >=0pt, shorten <=0pt] (0) to (1);
      \draw[arrow,shorten >=0pt, shorten <=0pt] (m2) to (m1);
      \draw[arrow,shorten >=0pt, shorten <=0pt] (dots) to  (m2);

      \path (0.north west) ++(-0.1,0.1) node[name=a1]{} (0.south east) ++(0.1,-0.1) 
        node[name=a2]{};
      \draw[rounded corners=6pt] (a1) rectangle (a2);

      \path (1.north west) ++(-0.1,0.1) node[name=a1]{} (1.south east) ++(0.1,-0.1) 
        node[name=a2]{};
      \draw[rounded corners=6pt] (a1) rectangle (a2);

      \path (m2.north west) ++(-1.2,0.1) node[name=z1]{} (m1.south east) ++(0.1,-0.1) 
        node[name=z2]{};
      \draw[rounded corners=6pt] (z1) rectangle (z2);
}
\par}

\medskip
\noindent
We fix the inputs of Algorithm~\ref{algo:explicit_sound} to be \(\sigma_i=\varnothing\) and  the initial preorder \(R_i\) is given by \(R_i(0) = \{0\}\), \(R_i(1) = \{0,1\}\), and \(R_i(n) = \interval[open left]{-\infty}{-1}\) for \(n < 0\). 
  The simulation preorder is therefore: \(\Rsim(0) = \{0\}\), \(\Rsim(1) = \{0,1\}\), and \(\Rsim(n) = \interval[open left]{-\infty}{n}\) for each \(n>0\). 
  Notice that \(\Rsim\) has infinitely many principals.
    Executing Algorithm~\ref{algo:explicit_sound}, we get \(\sigma = \{0,1\}\) after two \emph{Search} iterations.
  At this point, \(V = \varnothing\) and \(U = \varnothing\) holds, and the algorithm terminates and returns the correct result.

\noindent
  Consider instead the process of refining each principal \(R_i(x)\) such that \(R_i(x) \neq \Rsim(x)\).
  This process, which converges to \(\Rsim\), does not terminate after finitely many steps since infinitely many principals need to be refined.\lipicsEnd
\end{example}

\subsection{An Algorithm for the Reachable Simulation Partition}\label{subsec:partition}

Since, for finite state systems, it is possible to compute precisely the reachable blocks of the simulation partition in a na\"{\i}ve way (by first computing the simulation partition blocks, and then filtering out the ones containing no reachable state), it is a natural question to ask whether Algorithm~\ref{algo:explicit_sound} can be modified to compute precisely those blocks.
We designed a modified version of Algorithm~\ref{algo:explicit_sound} that computes precisely the reachable blocks of \(\Psim\), i.e., the output of this modified algorithm turns the sound containment \eqref{corr-stm-expl_sound2} of Theorem~\ref{thm:explicit_sound} into an equality. 
Intuitively, our modification changes the notions of reachability of principals for the sets \(U\) and \(V\), switching from that defined in~\eqref{eq:principals_intersect} to the one given by~\eqref{eq:principals_generated}.
Moreover, a third branch is added to the \(\KwSty{nif}\) statement of Algorithm~\ref{algo:explicit_sound}, which computes all the one-step successors of states in \(\sigma\) and checks whether the whole \(\post^*(I)\) set has been computed. 
Like Algorithm~\ref{algo:explicit_sound}, this modified version is guaranteed to terminate on finite state systems.

Let us recall that, due to the results given in Section~\ref{sec:complexity_finite}, computing the whole set \(\post^*(I)\) is, in general, unavoidable.
In particular, several features of this modified algorithm suggest that it terminates less often for infinite state systems compared to Algorithm~\ref{algo:explicit_sound}, because, 
intuitively, it relies more tightly on explicitly computing reachable states (e.g., some part of the algorithm explicitly relies on having computed the whole set \(\post^*(I)\)).
Therefore, we made the choice to focus  on Algorithm~\ref{algo:explicit_sound} in the main body of the paper, while
in Appendix~\ref{sec:completeonline} we provide the full modified algorithm as  Algorithm~\ref{algo:nonterminating_ex}, together with a more in-depth discussion of its features and formal correctness and termination results.

\section{2PR Triples for Designing a Symbolic Algorithm}\label{sec:2pr}\label{sec:symbolic_algo}
Symbolic approaches for simulation algorithms based on state partitions are known to be beneficial for algorithms manipulating infinite state systems, as shown by Henzinger et al.\@'s symbolic simulation algorithm for infinite graphs and, in particular, hybrid automata~\cite{henzingerComputingSimulationsFinite1995}, and have been proven to be advantageous in terms of space and time efficiency also for finite state systems \cite{Cece17, crt2011,Ranzato2013}.
Accordingly, we introduce the notion of 2-partitions-relation triple (2PR), generalizing the partition-relation pairs used in the most efficient symbolic simulation algorithms  \cite{Cece17,gpp03,ranzatoEfficientSimulationAlgorithm2010} as symbolic representation of a relation between system states.
We exploit here 2PRs to design a symbolic version of Algorithm~\ref{algo:explicit_sound}.
The rationale behind the need for 2PRs rather than partition-relation pairs (viz.\ 1PR) has more to do with enhancing the presentation and ease of understanding and less to do with fundamental limitations of partition-relation pairs.

\begin{definition}[\textbf{2PR Triple}]
\rm
Given a (possibly infinite) set \(\Sigma\), a triple \(\tuple{P, \tau, Q}\) with \(P, Q \in \Part(\Sigma)\) and \(\tau \colon P \rightarrow \wp(Q)\), is a \emph{2-partitions-relation} (2PR) \emph{triple}. \lipicsEnd
\end{definition}

A relation \(R \in \Rel(\Sigma)\) induces a 2PR triple \(\tuple{P_R, \tau_R, Q_R}\) as follows: 
\begin{align*}
 P_R &\ud \{y \in \Sigma \mid R(x) = R(y)\}_{x \in \Sigma} \enspace, & \tau_R(B) &\ud \{C \in Q_R \mid C \subseteq R(B)\}\enspace,\\
 Q_R &\ud \{y \in \Sigma \mid R^{-1}(x) = R^{-1}(y)\}_{x \in \Sigma} \enspace. 
\end{align*}
Hence, the rationale for maintaining two separate partitions \(P\) and \(Q\) is that \(P\) keeps track of states having the same principal, while \(Q\) keeps track of states occurring in the same set of principals.

Conversely, a 2PR triple \(\tuple{P, \tau, Q}\) encodes a relation \(R_{\tuple{P, \tau, Q}} \in \Rel(\Sigma)\) defined as \(R_{\tuple{P, \tau, Q}}(x) \ud \cup\tau(P(x))\) (the union of the blocks returned by \(\tau\)), and a partition \(P_{\tuple{P, \tau, Q}} \in \Part(\Sigma)\) given by \(P_{\tuple{P, \tau, Q}}\ud \{y \in \Sigma \mid \cup\tau(P(y)) = \cup\tau(P(x))\}_{x\in \Sigma}\), such that \(P_{\tuple{P, \tau, Q}}\) turns out to be coarser than \(P\). 
Some additional properties of 2PR triples can be found in Appendix~\ref{sec:2prproperties}.

We put forward Algorithm~\ref{algo:symbolic_sound}, designed as a revision of Algorithm~\ref{algo:explicit_sound} based on representing 
the  current relation \(R\) as a 2PR triple.
Algorithm~\ref{algo:symbolic_sound} is in symbolic logical form, in the sense that it symbolically represents and processes state relations as 2PR triples \(\tuple{P, \tau, Q}\). Since the refinement process preserves reflexivity of the underlying relation (as it happens for Algorithm~\ref{algo:explicit_sound}), it follows that \(\lambda B\in P. \cup\!\tau(B)\) is an extensive operator.
During its execution, for each state \(x\), \(R_{\tuple{P, \tau, Q}}(x)\) is a set of states candidates to simulate \(x\), and all the states in the same block of \(P_{\tuple{P, \tau, Q}}(x)\) are candidates to be simulation equivalent.

\begin{algorithm}[tbh]
\caption{\textsf{2PR-based Algorithm}}\label{algo:symbolic_sound}
{ \color{green!40!black}
\KwIn{A transition system \(G=(\Sigma,I,L,\mathord{\sra})\), an initial preorder \(R_i\in \PreO(\Sigma)\), an initial finite set \(\sigma_i \subseteq \post^*(I)\)}}
\(\Part(\Sigma) \ni P := \{y \in \Sigma \mid R_i(x) = R_i(y)\}_{x \in \Sigma}\)\;\label{loc:sym_P_defin}
\(\Part(\Sigma) \ni Q := \{y \in \Sigma \mid R_i^{-1}(x) = R_i^{-1}(y)\}_{x \in \Sigma}\)\;\label{loc:sym_Q_defin}
\lForAll{$B\in P$}{\(\wp(Q)\ni \tau(B):=\{C \in Q \mid C \subseteq R_i(B)\}\)\label{loc:sym_tau_defin}} 
\(\wp(\Sigma)\ni \sigma:=\sigma_i\)\; 
\While(){\True}{ 
 { \color{blue!75!black}
  \tcp{{\rm \textsc{Inv\textsubscript{1}}}: \(\!\!\forall x\in \Sigma\ldotp\! \Rsim(x) \subseteq\! R_{\tuple{P, \tau, Q}}(x)\subseteq\! R_i(x)\)}
  \tcp{{\rm \textsc{Inv\textsubscript{2}}}: \(\!\!\sigma_i \subseteq \sigma \subseteq \post^*(I)\)}
  \tcp{{\rm \textsc{Inv\textsubscript{3}}}: \(\!\!\forall B\in P\ldotp B\subseteq \cup\tau(B)\)}
  }
\(U:= \{B\in P \mid \cup\tau(B) \cap \sigma = \varnothing,\, \cup\tau(B) \cap (I\cup \post(\sigma))\neq \varnothing\)\};

\(V:= \{\tuple{a,B,C} \in L \times P^2 \mid \cup\tau(B) \cap \sigma \neq \varnothing,\) 
\hspace*{11.35em}\(B \cap \pre_a(C) \neq \varnothing, \cup\tau(B) \not\subseteq 
\pre_a(\cup\tau(C))\}\)\;

\Nif{\label{loc:sym_nif_begin}
	\((U \neq \varnothing) \longrightarrow \mathit{Search}:\) \label{loc:sym_search_begin}\\ 
	\Indp
		\Choose \(B\in U,\, s \in  (\cup\tau(B) \cap (I\cup \post(\sigma)))\)\;
		\(\sigma:= \sigma \cup \{s\}\)\;\label{loc:sym_sigmaupdate} \label{loc:sym_search_end}
	\Indm

  \((V \neq \varnothing) \longrightarrow \mathit{Refine}:\) \label{loc:sym_refine_begin}\\   
	\Indp
		\Choose \(\tuple{a, B, C}\in V\); \(S := \pre_a(\cup\tau(C))\)\;
		\(B' := B \cap \pre_a(C)\); \(B'' := B \smallsetminus \pre_a(C)\)\;\label{loc:sym_refine_splitB_begin}
		\(P.replace(B, \{B', B''\})\)\;
		\(\tau(B') := \tau(B)\); \(\tau(B'') := \tau(B)\)\;\label{loc:sym_refine_splitB_end}
		\ForAll{\(X \in \{E \in \tau(B') \mid E \cap S \neq \varnothing, E \not\subseteq S\}\)\label{loc:sym_refine_Q_begin}}{
			\(Q.\mathit{replace}(X, \{X \cap S, X \smallsetminus S\})\)\;
      \lForEach{\(A\in P\)}{\(\tau(A).\mathit{replace}(X, \{X \cap S, X \smallsetminus S\})\)}
		}\label{loc:sym_refine_Q_end}
		\( \tau(B') := \{E \in \tau(B') \mid E \subseteq S\}\)\;\label{loc:sym_refine_principal} \label{loc:sym_refine_end}
	\Indm
	
  \((U = \varnothing \land V = \varnothing) \longrightarrow {}\)\Return{\(\tuple{P, \tau, Q, \sigma}\)}\label{loc:sym_return}\;
}\label{loc:sym_nif_end}

}
\end{algorithm}

Following the approach of Algorithm~\ref{algo:explicit_sound}, this symbolic procedure computes a set of principals \(\{\cup\tau(B) \mid \cup\tau(B) \cap \sigma \neq \varnothing\}_{B\in P}\) which are known to contain a reachable state. 
A block \(B\) is in \(U\) if \(\cup\tau(B)\) does not contain states that are currently known to be reachable, while
it contains either an initial state or a successor of a state which is known to be reachable.
Also, the set \(V\) contains unstable symbolic triples, that is: \(\tuple{a,B,C}\) is in \(V\) if{}f \(\cup\tau(B)\) is known to contain some reachable state and there exists \(b \in B, c \in C\) such that \(R_{\tuple{P, \tau, Q}}(b)\) is \(a\)-unstable w.r.t.\ \(R_{\tuple{P, \tau, Q}}(c)\).
Algorithm~\ref{algo:symbolic_sound} either updates the reachability information for some principal of a block in \(U\) executing the \emph{Search} procedure 
or stabilizes the pair of blocks associated to some triple in \(V\) executing the \emph{Refine} procedure.
The \emph{Refine} procedure \(a\)-stabilizes a pair of blocks \((B, C)\) by possibly splitting \(B\) into \(B \cap \pre_a(C)\) and \(B \smallsetminus \pre_a(C)\) (lines~\ref{loc:sym_refine_splitB_begin}--\ref{loc:sym_refine_splitB_end}), and then  refines the principal \(\cup\tau(B \cap \pre_a(C))\) at lines~\ref{loc:sym_refine_Q_begin}--\ref{loc:sym_refine_Q_end} by first splitting blocks of \(Q\) if they are known to contain states occurring in different sets of principals (for the current underlying relation \(R_{\tuple{P, \tau, Q}}\)), and, successively, by removing at line~\ref{loc:sym_refine_principal} all the blocks not contained in \(\pre_a(\cup\tau(C))\).

\subparagraph*{A Region Algebra based Implementation.}
To implement this algorithm, we must identify a suitable representation for blocks of states.
In particular, what we need is a \emph{region algebra}, in the same vein of those used by Henzinger et al.~\cite{henzingerClassificationSymbolicTransition2005}, consisting of a set \(\mathcal{R}\) of regions, endowed with the operations listed below, and an extension (or interpretation) function \(\exten{\cdot}\colon\mathcal{R}\sra\wp(\Sigma)\) mapping each region to a set of contained states.
Encoding the input \(R_i\) as a 2PR triple over regions consistently with lines~\ref{loc:sym_P_defin}--\ref{loc:sym_tau_defin}, allows us to identify adequate region algebras 
following the algorithm's structure, that is, we need region algebras satisfying the following requirements:
\begin{enumerate}
\item Given \(s \in \Sigma\) (or, at least, \(s \in \post^*(I)\)), and \(r \in \mathcal{R}\), decide if \(\post(s) \cap \exten{r} \neq \varnothing\) holds and return an element in the intersection unless it is empty;
\item For each region \(r\) and \(a \in L\), compute a region \(\pre_a(r)\) such that \(\exten{\pre_a(r)} = \pre_a(\exten{r})\);
\item For each pair of regions \(r, r'\), compute a region \(r\cap r'\) such that \(\exten{r\cap r'} = \exten{r} \cap \exten{r'}\),
and a region \(r \smallsetminus r'\) such that \(\exten{r\smallsetminus r'} = \exten{r} \smallsetminus \exten{r'}\);
\item For every region \(r\), decide if \(\exten{r} = \varnothing\) holds.
\end{enumerate}
Once such a region algebra is defined, the input to Algorithm~\ref{algo:symbolic_sound} will be: the region algebra (composed of the operators listed in points 1--4 above), the initial set \(\sigma_i\), and a 2PR triple satisfying the constraints given by lines~\ref{loc:sym_P_defin}--\ref{loc:sym_tau_defin} of the algorithm, encoding the initial preorder.
The logical structure of the algorithm is then implemented using the region algebra operators in input.

\medskip
Upon termination of Algorithm~\ref{algo:symbolic_sound}, it turns out that \(R_{\tuple{P, \tau, Q}}^\sigma\) coincides with the set of reachable principals of \(\Rsim\) (cf.~\eqref{corr-stm-sym_sound} below), while for the reachable blocks of \(\Psim\), it turns out that we obtain a sound overapproximation of \(\Psim^{\post^*(I)}\) (cf.~\eqref{corr-stm-sym_sound2} below).  

\begin{restatable}[\textbf{Correctness of Algorithm~\ref{algo:symbolic_sound}}]{theorem}{correctnessSymbolic}
\label{thm:symbolic_sound}
  Let $\tuple{P, \tau, Q, \sigma}$ be the output of Algorithm~\ref{algo:symbolic_sound} on input \(G\) with \(|\Sigma|\in\bN\), $R_i\in \PreO(\Sigma)$ and \(\sigma_i \subseteq \post^*(I)\).
  Let $\Rsim\in \PreO(\Sigma)$ and $\Psim\in \Part(\Sigma)$ be, resp., the simulation preorder and partition w.r.t.\ $R_i$. 
  Then:
  \begin{align}
    \Rsim^{\post^*(I)} &=
    R_{\tuple{P, \tau, Q}}^\sigma\enspace,\tag{3.a}\label{corr-stm-sym_sound}\\
    \Psim^{\post^*(I)} &\subseteq 
     \{B \in P_{\tuple{P, \tau, Q}} \mid \exists E \in P. \: E \subseteq B 
      \wedge (\cup \tau(E))\cap \sigma \neq \varnothing\}\enspace .
     \tag{3.b}\label{corr-stm-sym_sound2}
  \end{align} 
\end{restatable}

Moreover, we are able to extend the correctness proofs for Algorithms~\ref{algo:explicit_sound}~and~\ref{algo:symbolic_sound} to infinite transition systems, under reasonable assumptions. Due to lack of space, the details of this extension are discussed in Appendix~\ref{sec:appendix_extension}.

\begin{restatable}[\textbf{Termination of Algorithm~\ref{algo:symbolic_sound}}]{theorem}{terminationSymbolic}
\label{thm:algo_sym_sound_finite_termination}
  Let \(G = (\Sigma, I,L, \sra)\) with \(|\Sigma|\in\bN\), \(R_i\in \PreO(\Sigma)\) and \(\sigma_i \subseteq \post^*(I)\). Then, Algorithm~\ref{algo:symbolic_sound} terminates on input \(G\), \(R_i\), and \(\sigma_i\).
\end{restatable}

As it is the case for the explicit Algorithm~\ref{algo:explicit_sound}, termination is guaranteed for finite state systems, while Algorithm~\ref{algo:symbolic_sound} may terminate on infinite state systems as well, as shown below in Example~\ref{ex:symbolic_infinite_cont}.
Moreover, the use of 2PR triples as a representation structure brings benefits in terms of termination on infinite systems, since there exist inputs on which Algorithm~\ref{algo:explicit_sound} does not terminate but Algorithm~\ref{algo:symbolic_sound} does. 
For instance, Algorithm~\ref{algo:explicit_sound} terminates for no input where \(R_i\) is such that the set of principals given by \(\{R_i(x) \mid x\in \Sigma, \: \Rsim(x) \cap \post^*(I) \neq \varnothing,\, \Rsim(x) \neq R_i(x)\}\) is infinite.
On such inputs Algorithm~\ref{algo:explicit_sound} executes the \emph{Refine} block infinitely often, as shown in
the following example. 

\begin{example}\label{ex:symbolic_infinite}
Consider the following infinite transition system, where \(\Sigma = \bN\), \(I=\{1\}\), \(L\) is a singleton, \(\sigma_i=I\), and the initial preorder \(R_i\) is \(\bN \times \bN\). 

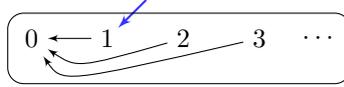
{\centering
\tikz[baseline=(1.south),auto]{
      \tikzstyle{arrow}=[->,>=latex']
      \path      
        (-1,0) node[name=0]{0}
        (0,0) node[name=1]{1}
        (1,0) node[name=2]{2}
        (2,0) node[name=3]{3}
        (2.8,0) node[name=dots]{$\cdots$}
        ;
      \node (C1) [above right of = 1] {};
      \draw[color=blue!80,thick,arrow={-Stealth},shorten >=-3pt, shorten <=3pt] (C1) to (1);
      
      \draw[arrow,shorten >=0pt,shorten <=0pt,] (1) to (0);
      \draw[arrow,shorten >=0pt, shorten <=0pt] (2) .. controls (-.4,-.4) .. (0);
      \draw[arrow,shorten >=0pt, shorten <=0pt] (3) .. controls (-.6,-.6) .. (0);
      
      \path (0.north west) ++(-0.1,0.1) node[name=w1]{} (2.8,-.2) ++(0.4,-0.4) node[name=w2]{};
      \draw[rounded corners=6pt] (w1) rectangle (w2);
}
\par}

\medskip
\noindent
The simulation preorder \(\Rsim\) is such that \(\Rsim(0) = \bN\), and \(\Rsim(n) = \bN \smallsetminus \{0\}\) for \(n \geq 1\), meaning that \(\Psim = \{\{0\}, \{1, 2, \dots\}\}\).
During execution of Algorithm~\ref{algo:explicit_sound}, the set \(U\) is empty, since the state \(1\) is in every principal.
On the other hand, the set \(V\) contains a triple \(\tuple{a, i, 0}\) for every positive integer \(i>0\).
Each time a \emph{Refine} iteration is executed, exactly one principal is updated. Assume that a triple \(\tuple{a, i, 0}\in V\) is selected, so that the refinement \(R(i) = \bN \smallsetminus \{0\}\) will be performed.
As a consequence, the execution will never terminate because \(V\) initially contains infinitely many triples, and at each iteration exactly  one of them will be removed from $V$. \lipicsEnd
\end{example}

On the other hand, Algorithm~\ref{algo:symbolic_sound} is able to refine infinitely many principals of the underlying relation in a single step of \emph{Refine}. To illustrate this, we consider the input from Example~\ref{ex:symbolic_infinite} and show that Algorithm~\ref{algo:symbolic_sound} converges in a few iterations. %

\begin{example}\label{ex:symbolic_infinite_cont}
Consider the same input as in Example~\ref{ex:symbolic_infinite}. Here, Algorithm~\ref{algo:symbolic_sound} terminates.
The initial 2PR triple is given by \(P = Q = \{\bN\}\) and \(\tau(\bN) = \{\bN\}\), the set \(U\) is empty, and the set \(V\) is \(\{\tuple{a, \bN, \bN}\}\).
Executing one iteration refines both \(P\), \(Q\) and \(\tau\), so that \(P = Q = \{\{0\}, \bN \smallsetminus \{0\}\}\), \(\tau(\{0\}) = P\) and \(\tau(\bN \smallsetminus \{0\}) = \{\bN \smallsetminus \{0\}\}\).
At this point, the underlying relation is such that \(R_{\tuple{P, \tau, Q}} = \Rsim\), so that the execution of Algorithm~\ref{algo:symbolic_sound} terminates. \lipicsEnd
\end{example}

Let us finally remark that Algorithm~\ref{algo:symbolic_sound} initializes the partitions \(P\) and \(Q\) of the 2PR triple in such a way that \(P\) keeps track of states having the same principal for \(R_i\), and \(Q\) keeps track of states occurring in the same set of principals for \(R_i\).
The idea is to start with \(P\) and \(Q\) as coarse as possible.
On the opposite end, \(P\) and \(Q\) can be initialized with \(P=Q\) such that each of their blocks is a singleton. 
Algorithm~\ref{algo:symbolic_sound} with that initialization remains correct and is not much different from Algorithm~\ref{algo:explicit_sound} since the relation \(R\) of Algorithm~\ref{algo:explicit_sound} and \(\tau\) of Algorithm~\ref{algo:symbolic_sound} are essentially the same.    

Coming back to the goal of having \(P\) and \(Q\) as coarse as possible, we next show that when Algorithm~\ref{algo:symbolic_sound} executes, 
each refinement of \(P\) (or \(Q\)) can be attributed to a change among states having the same principal for the current \(R_{\tuple{P, \tau, Q}}\) (or a change among states occurring in the same set of principals for the current \(R_{\tuple{P, \tau, Q}}\)). 
This is formally stated as follows.

\begin{restatable}[]{property}{algosymSplittingStates}
\label{pr:algosymSplittingStates}
Let \(\tuple{P, \tau, Q, \sigma}\) be the output of Algorithm~\ref{algo:symbolic_sound} with input \(G\), \(R_i \in \PreO(\Sigma)\) and \(\sigma_i \subseteq \post^*(I)\), and assume that the execution terminates after \(t\) iterations. Moreover, let \(x, y \in \Sigma\) be two initially equivalent states (i.e., \(R_i(x) = R_i(y)\)), and \(\tuple{P_j, \tau_j, Q_j}\) be the 2PR triple after \(j = 0,\dots,t\) iterations of the algorithm, and \(R_j = R_{\tuple{P_j, \tau_j, Q_j}}\) be the corresponding underlying relation. Then:
\begin{gather}
  P(x) \neq P(y) \Ra R_j(x) \neq R_j(y) \;\;\text{ for some } j \in \interval{0}{t}\enspace, \label{eq:2prsplitP}
  \shortintertext{while a dual property for \(Q\) holds:}
  Q(x) \neq Q(y) \Ra R^{-1}_j(x) \neq R^{-1}_j(y) \;\;\text{ for some } j \in \interval{0}{t}\enspace. \label{eq:2prsplitQ}
\end{gather}

\end{restatable}

We have observed, through a prototype implementation in Python of our algorithms and a corresponding experimental evaluation, that employing 2PR triples brings benefits in terms of execution speed, in practice. Due to lack of space, we report in Appendix~\ref{sec:experiments} on such empirical evidence over a set of benchmarks consisting of well-known mutual exclusion protocols.

\section{Conclusion and Future Work}\label{sec:conclusion}

We introduced and proved the correctness of algorithms solving the reachable simulation problem, which discloses some fundamental differences in decidability and complexity w.r.t.\ to the analogous problem for 
bisimulation studied by Lee and Yannakakis~\cite{leeOnlineMinimizationTransition1992}. 
To the best of our knowledge, this is the first investigation of Lee and Yannakakis' problem recast to the simulation preorder and partition.   
Algorithm~\ref{algo:symbolic_sound} is the most relevant one for practical purposes, since this procedure converges 
on all finite state systems and is well-suited to handle infinite state systems thanks to its symbolic representation, 
being able to converge on some---but not all, due to an undecidability result---such infinite  systems.

Future work includes, but is not limited to, the following tasks:
\begin{itemize}
  \item To further investigate the reachable simulation problem for transition systems with infinitely many states. One such example of infinite systems is given by timed automata or hybrid systems since they enjoy subclasses with finite similarity quotients.
  \item To establish upper bounds of the algorithms in terms of number of basic operations performed on the region algebra, following an approach similar to Lee and Yannakakis~\cite[Section~3]{leeOnlineMinimizationTransition1992}.
  \item To study algorithmic improvements of our algorithms, for instance, by leveraging the data structures such as queue and stacks used by Lee and Yannakakis~\cite{leeOnlineMinimizationTransition1992} algorithm for bisimulation. 
  \item To investigate the reachable  problem, analogous to Problem~\ref{problem1}, for further behavioral relations. One appealing relation is branching bisimilarity for labeled transition systems or stuttering equivalence for Kripke structures, which have been  proven in
  \cite{groote2017,groote2020} to be symbolically computable in $O(|{\sra}|\log |\Sigma|)$ time, the same complexity of the renowned Paige-Tarjan bisimulation algorithm~\cite{pt87}.
\end{itemize}

\clearpage
\appendix
\section{Auxiliary Algorithms}
\begin{algorithm}[H]
\caption{\textsf{Sim}}\label{basicsim}
{ \color{green!40!black}
\KwIn{A transition system \(G=(\Sigma,I,L,\mathord{\sra})\) , a relation \(R_i \in \Rel(\Sigma)\) s.t. \(\Rsim \subseteq R_i \subseteq R_{\text{po}}\) for some preorder \(R_{\text{po}}\), and the simulation preorder \(\Rsim\) induced by \(R_{\text{po}}\).} }
\(\Rel(\Sigma)\ni R:=R_i \)\;
\While{$\exists x, x'\!\in\! \Sigma, a\!\in\! L\ldotp x \sraa{a} x' \land R(x)\not\subseteq \pre_a(R(x'))$}{
 { \color{blue!75!black}
\tcp{{\rm \textsc{Inv\textsubscript{1}}}: \(\forall x\in \Sigma\ldotp \Rsim(x) \subseteq R(x)\subseteq R_i(x)\)}
}
$R(x):=R(x) \cap \pre_a(R(x'))$\;
}

\Return{$R$}\;
\end{algorithm}

\begin{assumption}\label{assumption_algo4}
  If \(y \in R_i(x) \smallsetminus \Rsim(x)\) for some \(x \in \Sigma\), then there exists an execution (not necessarily terminating) of {\rm \textsf{Sim}} over input \(R_i\) such that, after a finite number of iterations, \(y \not\in R(x)\) holds.
\end{assumption}

\subsection{A Refined Algorithm}
\begin{algorithm}[H]
\caption{\textsf{Refined Algorithm} (\(\RefAlgo\))}\label{algo6_refine_ex}
{ \color{green!40!black}
\KwIn{A transition system \(G=(\Sigma,I,L,\mathord{\sra})\), the set of reachable states \(\post^*(I)\), and a relation \(R_i\in\Rel(\Sigma)\) s.t. \(\Rsim \subseteq R_i \subseteq R_{\text{po}}\), where \(R_{\text{po}}\) is some preorder and \(\Rsim\) is the simulation preorder induced by \(R_{\text{po}}\).} }
\(\Rel(\Sigma)\ni R:=R_i \)\;
\(\wp(\Sigma)\ni \sigma:=\post^*(I)\)\; 
\While(){\True}{ 
 { \color{blue!75!black}
  \tcp{{\rm \textsc{Inv\textsubscript{1}}}: \(\forall x\ldotp \Rsim(x) \subseteq R(x)\subseteq R_i(x)\)}
  \tcp{{\rm \textsc{Inv\textsubscript{2}}}: \(\forall x\in \Sigma\ldotp x\in R(x)\)}
  }
\(V:= \{\tuple{a,x,x'} \in L \times \Sigma^2 \mid R(x) \cap \sigma \neq \varnothing, x\sraa{a} x', R(x) \not\subseteq  \pre_a(R(x'))\}\)\;

\eIf{\((V \neq \varnothing)\)}{
	\(\mathit{Refine}:\) \\  
	\Indp
		\Choose \(\tuple{a, x,x'}\in V\)\;
		\( R(x) := R(x) \cap \pre_a(R(x'))\)\;
	\Indm
}{
	\Return{$\tuple{R,\sigma}$}\;  
}

}

\end{algorithm}

\medskip
Let us observe that if Algorithm~\ref{algo:explicit_sound} is called with \(\sigma_i=\post^*(I)\), then the set \(U\) will be empty at each iteration, so that \emph{Search} never executes.
Based on this observation, we define Algorithm~\ref{algo6_refine_ex}, also denoted by \(\RefAlgo\), obtained as 
partial evaluation of Algorithm~\ref{algo:explicit_sound} by assuming to have in input \(\sigma_i=\post^*(I)\). 
Moreover, Algorithm~\ref{algo6_refine_ex} differs from Algorithm~\ref{algo:explicit_sound} in that we require the input relation \(R_i\) to be reflexive but not necessarily transitive.
The rationale is that, Algorithm~\ref{algo6_refine_ex} will be employed as a subroutine having in input a relation \(R_i\) that sits in between a preorder \(R_{\text{po}}\) and the simulation preorder \(\Rsim\) induced by \(R_{\text{po}}\).    
More precisely, it turns out that for Algorithm~\ref{algo6_refine_ex} the invariant \({\rm \textsc{Inv\textsubscript{1}}}\ud {\Rsim \subseteq R_i \subseteq R_{\text{po}}}\) holds.
Note that \(R_i\) is reflexive since \(\Rsim\) is but \(R_i\) need not be transitive.
Under these assumptions on the inputs, it is easily seen that Algorithm~\ref{algo6_refine_ex} inherits correctness and termination results as given by Theorems~\ref{thm:explicit_sound}~and~\ref{thm:algo_sound_finite_termination}.

\section{An Algorithm for the Reachable Simulation Partition}\label{sec:completeonline}

\begin{algorithm}[H]
\caption{\textsf{Reachable Simulation Partition Algorithm}}\label{algo:nonterminating_ex}
{ \color{green!40!black}
\KwIn{A trans.\ system \(G\!=\!(\Sigma,I,L,\mathord{\sra})\), an initial preorder \(R_i\in \PreO(\Sigma)\), an initial finite set \(I \subseteq \sigma_i \subseteq \post^*(I)\).} }
\(\Rel(\Sigma) \ni R := R_i\)\;

\(\wp(\Sigma)\ni \sigma:=\sigma_i\)\; 
\(\wp(U)\ni \Ubad:=\varnothing\)\; 
\While(){\True}{ 
 { \color{blue!75!black}
  \tcp{{\rm \textsc{Inv\textsubscript{1}}}: \(\forall x\in \Sigma\ldotp \Rsim(x) \subseteq R(x)\subseteq R_i(x)\)}
  \tcp{{\rm \textsc{Inv\textsubscript{2}}}: \(\sigma_i \subseteq \sigma \subseteq \post^*(I)\)}
  \tcp{{\rm \textsc{Inv\textsubscript{3}}}: \(\forall x\in\Sigma\ldotp x\in R(x)\)}
  \tcp{{\rm \textsc{Inv\textsubscript{4}}}: \(\Ubad \subseteq U\)}
  }
\(U:= \{x \in \Sigma \mid \nexists s \in \sigma \ldotp R(x) = R(s),\, R(x)\cap \post(\sigma)\neq \varnothing\}\)\;\label{loc:ntex_udef}
\(V:= \{(a,x,x')\in L {\times} \Sigma^2 \mid \exists s {\in} \sigma \ldotp R(x) {=} R(s),\, (x\sraa{a} x', R(x) \nsubseteq \pre_a(R(x')) \}\)\;\label{loc:ntex_vdef}

\Nif{
	\((U\smallsetminus \Ubad \neq \varnothing) \longrightarrow \mathit{Search}:\)\\ 
	\Indp
		\Choose \(x\in U \smallsetminus \Ubad\)\;
		\(S := (R(x)\cap \post(\sigma)) \smallsetminus \sigma\)\;
		\eIf{\(S \neq \varnothing\)}{
			\Choose \(s \in S\)\;
			\(\sigma:= \sigma \cup \{s\}\)\;\label{loc:ntex_sigmaupdate}
			\(\Ubad := \varnothing\)\;\label{loc:ntex_resetubad1} 
		}{
			\(\Ubad := \Ubad \cup \{x\}\)\;\label{loc:ntex_ubadupdate}
		}
	\Indm

	\((V \neq \varnothing) \longrightarrow \mathit{Refine}:\)\\  
	\Indp
		\Choose \((a,x,x')\in V\)\;
		\(S:=\pre_a(R(x'))\)\;
		\(R(x):= R(x) \cap S\)\;
		\(\Ubad := \varnothing\)\;\label{loc:ntex_resetubad2} 

	\Indm
	
	\((U = \Ubad \neq \varnothing \land V = \varnothing) \longrightarrow \mathit{Expand}:\)\\ 	
	\Indp
		\eIf{\(\post(\sigma) \subseteq \sigma\)}{
			 { \color{blue!75!black} \tcp{{\rm Run Algo.~\ref{algo6_refine_ex} with input \(R\) and \(\sigma\)}}}
			\Return{$\RefAlgo(R, \sigma)$}\;\label{loc:ntex_algo2call} 
		}{
			\(\sigma:= \sigma \cup \post(\sigma)\)\;\label{loc:ntex_sigma_expand_update}
			\(\Ubad := \varnothing\)\;\label{loc:ntex_resetubad3} 
		}		
	\Indm
	
	\((U = \varnothing \land V = \varnothing) \longrightarrow {}\)\Return{$\tuple{R,\sigma}$}\;\label{loc:ntex_return2} 
}

}
\end{algorithm}

\medskip
Algorithm~\ref{algo:explicit_sound} is sound and complete for $\Rsim$ (cf.\ Theorem~\ref{thm:explicit_sound}~\eqref{corr-stm-expl_sound}) 
and, as shown in Example~\ref{ex:shortcoming}, is in general sound but not precise for $\Psim$ (cf.\ Theorem~\ref{thm:explicit_sound}~\eqref{corr-stm-expl_sound2}).
To achieve this precision, we perform several key changes to Algorithm~\ref{algo:explicit_sound} leading to the pseudocode in logical form 
presented as Algorithm~\ref{algo:nonterminating_ex}.
First, we modify the definitions of the sets $U$ and $V$ at lines~\ref{loc:ntex_udef} and~\ref{loc:ntex_vdef}.
Informally speaking, the changes made to the definitions of \(U\) and \(V\) reflect the difference between definitions \eqref{eq:principals_intersect} and \eqref{eq:principals_generated} of reachable principal.
Moreover, this algorithm keeps track of a subset of $U$, denoted by 
$\Ubad$, that contains states on which executing a \emph{Search} iteration can never expand $\sigma$. 
Furthermore, we added an \emph{Expand} guarded procedure in the loop of nondeterministic choices whose role is to guarantee progress by adding new elements to \(\sigma\).
The \emph{Expand} procedure at line~\ref{loc:ntex_algo2call} has a return statement invoking Algorithm~\ref{algo6_refine_ex} as a subroutine whenever \(\sigma\) turns out to be the whole set $\post^*(I)$ of reachable states.
It is worth recalling that the results in Section~\ref{sec:complexity_finite} entail that computing the whole 
set $\post^*(I)$ of reachable states is, in general, unavoidable.

The main feature of Algorithm~\ref{algo:nonterminating_ex} is that if it terminates with output $\tuple{R,\sigma}$, then it induces a partition whose blocks having a nonempty intersection with $\sigma$ are in a 1-1 correspondence with the reachable blocks of \(\Psim\), and they split every reachable block consistently with \(\Psim\) (cf.~\eqref{corr-stm-ntex2} below). 
Note that a block \(B\in P^\sigma\) might be strictly contained in the corresponding block \(\Psim(B)\) of \(\Psim\), but this algorithm ensures that \(\Psim(B) \smallsetminus B \cap \post^*(I) = \varnothing\) (cf.~\eqref{corr-stm-ntex3} below).
On the other hand, Algorithm~\ref{algo:nonterminating_ex} is also precise for the reachable principals of the simulation preorder $\Rsim$ where reachability for principals is defined according to definition~\eqref{eq:principals_generated} (cf.~\eqref{corr-stm-ntex} below). 

\begin{restatable}[\textbf{Correctness of Algorithm~\ref{algo:nonterminating_ex}}]{theorem}{correctnessNonterminatingExplicit}
\label{thm:CorrectnessAlgoNonterminatingEx}
  Let $\tuple{R,\sigma}\in \Rel(\Sigma)\times\wp(\Sigma)$ be the output of Algorithm~\ref{algo:nonterminating_ex} on input \(G\) with \(|\Sigma|\in\bN\), $R_i\in \PreO(\Sigma)$, \(I \subseteq \sigma_i \subseteq \post^*(I)\). 
  Let $\Rsim\in \PreO(\Sigma)$ and $\Psim\in \Part(\Sigma)$ be, resp., the simulation preorder and partition w.r.t.\ $R_i$. %
  Let \(P\ud \{y \in \Sigma \mid R(y) = R(x)\}_{x\in \Sigma}\in \Part(\Sigma)\).
  Then:
  \begin{align}
    &\{\Rsim(x) 
    \mid x\in \post^*(I)\} =
    \{R(x)\mid x\in \sigma\},\tag{3.a}\label{corr-stm-ntex}\\
    &\{B\cap\post^*(I) 
  \mid B\in\Psim^{\post^*(I)}\} = 
  \{B\cap\post^*(I)\mid B\in P^\sigma\},\tag{3.b}\label{corr-stm-ntex2}\\
  &\Psim^{\post^*(I)} =
  \{\Psim(B)\mid B\in P^\sigma \}\enspace.\tag{3.c}\label{corr-stm-ntex3}
\end{align} 
\end{restatable}

The following example shows the role 
of the \emph{Expand} guarded procedure and of the call of Algorithm~\ref{algo6_refine_ex} at line~\ref{loc:ntex_algo2call}
in ensuring termination.

\begin{example}\label{ex:ntex_algo2call_motivation}
  Consider the following finite system, where \(I=\{0\}\), \(L\) is a singleton, \(\sigma_i=I\), and the initial preorder \(R_i\) is given by \(R_i(0) = R_i(1) = \{0,1\}\), \(R_i(2) = \{2\}\) and \(R_i(3) = \{2,3\}\), where, as before, the boxes gather states sharing the same principal in \(R_i\).  
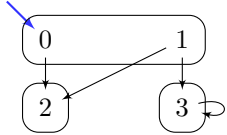
  \begin{center}
  \begin{tikzpicture}[scale=0.9]
    \tikzstyle{arrow}=[->,>=latex']
    \path      
      (1,0) node[name=0]{0}
      (3,0) node[name=1]{1}
      (1,-1) node[name=2]{2}
      (3,-1) node[name=3]{3}
      ;
    \node (C1) [above left of = 0] {};
    \draw[color=blue!80,thick,arrow={-Stealth},shorten >=-3pt, shorten <=3pt] (C1) to (0);
    \draw[arrow,shorten >=0pt, shorten <=0pt] (0) to (2);
    \draw[arrow,shorten >=0pt, shorten <=0pt] (1) to (3);
    \draw[arrow, shorten >=0pt, shorten <=0pt] (1) to  (2);
    \draw[arrow] (3) edge[loop right] node[right] {} (1);

    \path (0.north west) ++(-0.1,0.1) node[name=a1]{} (1.south east) ++(0.1,-0.1) 
      node[name=a2]{};
    \draw[rounded corners=6pt] (a1) rectangle (a2);

    \path (2.north west) ++(-0.1,0.1) node[name=z1]{} (2.south east) ++(0.1,-0.1) 
      node[name=z2]{};
    \draw[rounded corners=6pt] (z1) rectangle (z2);

    \path (3.north west) ++(-0.1,0.1) node[name=w1]{} (3.south east) ++(0.1,-0.1) 
      node[name=w2]{};
    \draw[rounded corners=6pt] (w1) rectangle (w2);

  \end{tikzpicture}    
\end{center}
The simulation preorder \(\Rsim\) is such that \(\Rsim(0) = \{0,1\}\), \(\Rsim(1) = \{1\}\), \(\Rsim(2) = \{2\}\) and \(\Rsim(3) = \{3\}\), meaning that states $0$ and $1$ are split in \(\Psim\).
Executing Algorithm~\ref{algo:nonterminating_ex}, we find that \(\sigma = \{0,2\}\) after the first \emph{Search} iteration.
At this point, \(V = \varnothing\) and \(U =\{3\}\), and the state \(3\) will be inserted into \(\Ubad\) after a further iteration
of  \emph{Search}.
Executing an \emph{Expand} procedure is then unable to enlarge \(\sigma\) since \(\post(\sigma) \subseteq \sigma\), so that 
the algorithm will execute the call $\RefAlgo(R, \sigma)$ of Algorithm~\ref{algo6_refine_ex} 
at line~\ref{loc:ntex_algo2call}, thus computing that \(R(1)= \{1\}\), and, in turn, 
separating the states \(0\) and \(1\) in $P$. \lipicsEnd
\end{example}

As aforementioned, it turns out that Algorithm~\ref{algo:nonterminating_ex} terminates on finite state systems.

\begin{restatable}[\textbf{Termination of Algorithm~\ref{algo:nonterminating_ex}}]{theorem}{terminationNonterminatingExplicit}
\label{thm:algo_complete_finite_termination}
  Let \(G = (\Sigma, I, L, \sra)\) with \(|\Sigma|\in\bN\), \(I \subseteq \sigma_i\subseteq\post^*(I)\) and \(R_i\in \PreO(\Sigma)\). Then, Algorithm~\ref{algo:nonterminating_ex}  terminates on input \(G\), \(R_i\) and \(\sigma_i\).
\end{restatable}

\section{Complexity Proof}\label{sec:app-complexity}

\nlhardnesslmembership*
\begin{proof}
We show the NL-hardness result by means of a reduction using the st-connectivity problem for leveled directed graphs.
A \emph{leveled directed graph} is a directed graph $G$ whose nodes
are partitioned in $k>0$ levels \(A_1,A_2,\ldots,A_k\), and every edge \((n,m)\) of $G$ is such that if \(n\in A_i\), 
for some \(0<i<k\), then \(m\in A_{i+1}\).
The st-connectivity problem in a leveled directed graph 
asks given two nodes \(s\in A_1\) and \(t\in A_k\) whether there exists a path in $G$ from \(s\) to \(t\).
The st-connectivity problem for leveled directed graphs is NL-complete \cite[p.~333]{sipserIntroductionTheoryComputation2006}.

Now, we reduce the st-connectivity problem to our decision Problem~\ref{rsdpfs}.
We first observe that given an instance of the st-connectivity problem in leveled directed graphs we can add self-loops to every node of the graph while preserving the existence or not of a path from \(s\) to \(t\).
Next, given such a leveled directed graph \(G=(N, E)\) with nodes \(N= A_1\cup A_2\cup\cdots\cup A_k\), edges \(E\), and two nodes \(s\in A_1\), \(t\in A_k\), define the transition system \(T_s\) by taking \(G\) with set of initial states \(I=\{s\}\).  
Next, consider as initial preorder relation \(R_i = \{ \{t\}\times N\} \cup \{ (N\smallsetminus\{t\}) \times (N\smallsetminus\{t\}) \}\).
It is easy to check, using a reasoning analogous to the undecidability proof in Section~\ref{sec:undecidability}, that \(R_i\) is the simulation preorder of \(T_s\) (that is, \(R_i=\Rsim\)). In particular, since \(G\) is a leveled directed graph, then \(t \sra x\) for no state \(x\), so that the only inclusions which need to be checked are 
\(R_i(N\smallsetminus \{t\}) \subseteq \pre(R_i(N\smallsetminus \{t\}))\),
\(R_i(N\smallsetminus \{t\}) \subseteq \pre(R_i(\{t\}))\),
and \(R_i(\{t\}) \subseteq \pre(R_i(\{t\}))\), and they all follow by definition of $T_s$ and $R_i$.
In turn, the simulation equivalence \(\Rsim\cap(\Rsim)^{-1}\) induces the partition \(\Psim=\{ N\smallsetminus\{t\}, \{t\} \}\).
Finally, it turns out that there exists a path from \(s\) to \(t\) in \(G\) if{}f all the blocks of \(\Psim\) are reachable.
More precisely, \((N\smallsetminus\{t\}) \cap \post^*(I) \neq \varnothing\) since \(I=\{s\}\) and \(s\in N\smallsetminus\{t\}\);
also, \(\{t\} \cap \post^*(I)\neq\emptyset\) if{}f there exists a path from \(s\) to \(t\).
Note that our reduction uses only two blocks as required.

Next, we show that the formulation of the decision Problem~\ref{rsdpfs} for bisimulation (rather than simulation) partition \(\Pbis\) consisting of two blocks \(B_0, B_1\) is in L.
Assume, without loss of generality, that \(I\subseteq B_0\) (the other cases are easily settled).
We are left to decide whether \(B_1\cap \post^*(I)\neq\varnothing\), which can be done by scanning one by one the transitions of the system to find a pair \((n,m)\) with \(n\in B_0\) and \(m\in B_1\).
Indeed, assume that \(B_1\cap \post^*(I)\neq\varnothing\) holds.
Therefore there exists a path from a state in \(I \subseteq B_0\) to a state in \(B_1\). 
Since the bisimulation partition \(\Pbis=\{B_0,B_1\}\), there exists a path in \(G\) from \(I\) to \(B_1\) if{}f every state in \(B_0\) has a transition into \(B_1\).
Thus to check whether \(B_1\cap \post^*(I)\neq\varnothing\) holds, it suffices to find a transition \(n\rightarrow m\) with \(n\in B_0\) and \(m\in B_1\). 
We need no more than logarithmic space to scan one by one the transitions and to prove the existence or absence of such an edge and we are done.
\end{proof}

\section{Properties of 2PR Triples}\label{sec:2prproperties}
\begin{restatable}[]{property}{twoprExact}
\label{pr:2prExact}
If \(R\in \Rel(\Sigma)\) (no assumption on $R$), then \(R = R_{\tuple{P_R, \tau_R, Q_R}}\) holds.
\end{restatable}

Moreover, 2PR triples encoding reflexive relations can be characterized by extensivity of  \(\lambda B\in P. \cup\!\tau(B)\):
\begin{restatable}[]{property}{twoprReflexiveCharacterization}
\label{pr:2prReflexiveCharacterization}
Let \(\tuple{P, \tau, Q}\) be a 2PR triple over \(\Sigma\). Then:
\[
R_{\tuple{P, \tau, Q}} \textit{ is reflexive} \Lra \forall B \in P \ldotp B \subseteq \cup\tau(B)\enspace .
\]
\end{restatable}
Finally, it turns out that 2PR triples induced by preorders enjoy the following property:
\begin{restatable}[]{property}{twoprQO}
\label{pr:2prQO}
Let \(R\) be a preorder and \(\tuple{P_R, \tau_R, Q_R}\) be the induced 2PR triple. Then, \(P_R = Q_R\).
\end{restatable}

\section{Proofs}\label{sec:app-proofs}
\correctnessExplicit*
\begin{proof}
The following facts on the output tuple $\tuple{R,\sigma}$ hold:
\begin{romanenumerate}
\item At every iteration, $R$ is reflexive.\\
This holds because at the beginning $R$ is a preorder and the update and refine block of $R$ at lines~\ref{loc:ex_refine_begin}--\ref{loc:ex_refine_end}
of Algorithm~\ref{algo:explicit_sound}
preserves the reflexivity of $R$, in particular the refinement statement at line \ref{loc:ex_refine_principal}.
\item For all $y\in \Sigma$, $\Rsim(y) \subseteq R(y)$.\\ 
This holds because the \emph{Refine} block of Algorithm~\ref{algo:explicit_sound} is always correct, so that  
$\Rsim(y) \subseteq R(y)\subseteq R_i(y)$ holds for every iteration of Algorithm~\ref{algo:explicit_sound}.

\item $x\in \post^*(I)$ $\Ra$ $R(x)\cap \sigma \neq \varnothing$.\\
This is proven by induction on $n\in \bN$ such that $I\sra^n x$. 
If $n=0$ then $x\in I$, so that, since $R$ is reflexive following \textcolor{lipicsGray}{\sffamily\bfseries\upshape (i)}, 
$x\in R(x)$ and therefore $R(x)\cap (I\cup \post(\sigma))\neq \varnothing$. 
Hence, $U=\varnothing$ implies that  $R(x)\cap \sigma \neq \varnothing$. 
If $n>0$ then $I\sra^n x' \sraa{a} x$, and
by inductive hypothesis, $R(x')\cap \sigma \neq \varnothing$. 
Therefore, since $V=\varnothing$,   $R(x') \subseteq \pre_a(R(x))$ must hold. 
Thus, $\pre_a(R(x)) \cap \sigma \neq \varnothing$, so that $R(x) \cap \post_a(\sigma) \neq \varnothing$. 
Hence, $U=\varnothing$ implies  $R(x)\cap \sigma \neq \varnothing$.

\item For all $x \in \Sigma$, $R(x)\cap \sigma \neq \varnothing \Ra R(x) = \Rsim(x)$.\\
Let $\Sim$ be the basic simulation algorithm as recalled in Algorithm~\ref{basicsim}.
By \textcolor{lipicsGray}{\sffamily\bfseries\upshape (ii)}, $\Rsim \subseteq R \subseteq R_i$.
Thus, $\Sim(R) = \Rsim$ because
$\Rsim = \Sim(\Rsim) \subseteq \Sim(R) \subseteq \Sim(R_i) = \Rsim$.
Assume, by contradiction, that $\bS \ud \{R(x) \in \wp(\Sigma) \mid R(x)\cap \sigma \neq \varnothing, R(x) \neq \Rsim(x)\}\neq \varnothing$,
so that each $R(x)\in \bS$, will be refined 
at some iteration of $\Sim(R)$.
Then, let \(R(x) \in \bS\) be the first principal in \(\bS\) such that  $R(x)$  is refined by $\Sim$
at some iteration $\mathit{it}$ whose current relation is $R'$. Thus, 
each principal \(R(z)\in\bS\) is such that $R(z)=R'(z)$ because no principal 
in $\bS$ was refined  by $\Sim$ before this iteration $\mathit{it}$.
Let $R'(y)\subseteq R(y)$ be the principal of $R'$ used by $\Sim$ in this iteration $\mathit{it}$ 
to refine \(R(x)\), so that \(x \sraa{a} y\),  
\(R'(x)=R(x) \not\subseteq \pre_a(R'(y))\).
We have $R(x)\cap \sigma \neq \varnothing$, \(x \sraa{a} y\) and $V=\varnothing$ entail $R(x)\subseteq \pre_a(R(y))$. Hence, 
$R(x)\cap \sigma \neq \varnothing$ implies  
$R(y)\cap \post_a(\sigma)\neq \varnothing$, 
so that $U=\varnothing$ entails $R(y)\cap \sigma\neq \varnothing$.
If $R(y)\not\in\bS$ then $R(y)\cap \sigma \neq \varnothing$ implies $R(y)=\Rsim(y)\subseteq
R'(y)\subseteq R(y)$, so that $R'(y)= R(y)$ holds. 
If $R(y)\in\bS$ then, since $R(x)$ is the first principal in $\bS$ to be refined by $\Sim$, we 
have that $R'(y)= R(y)$ holds. Thus, in both cases, $R'(y)= R(y)$ must hold, 
so that $R(x) \not\subseteq \pre_a(R(y))$. Moreover, $R(x)\cap \sigma \neq \varnothing$, 
\(x \sraa{a} y\) and $V=\varnothing$ imply $R(x) \subseteq \pre_a(R(y))$, 
which is therefore a contradiction. Thus, $\bS=\varnothing$, and, in turn, $R(x) = \Rsim(x)$.

\item $\sigma\subseteq \post^*(I)$. \\
This follows by \(\sigma_i \subseteq \post^*(I)\) and because $\sigma$ is updated at line~\ref{loc:ex_sigmaupdate} of Algorithm~\ref{algo:explicit_sound} by 
adding $s\in I\cup \post(\sigma)$ and $\post^*(I)$ is the least fixpoint
of $\lambda X.I\cup \post(X)$. 

\end{romanenumerate}
\noindent
Let us now show the equality \eqref{corr-stm-expl_sound}. 

$(\subseteq)$
Let $x\in \Sigma$ such that  $\Rsim(x)\cap \post^*(I)\neq \varnothing$. Thus, there exists $s\in \post^*(I)$
such that $s\in \Rsim(x)$. By \textcolor{lipicsGray}{\sffamily\bfseries\upshape (iii)}, $R(s) \cap \sigma \neq \varnothing$. By \textcolor{lipicsGray}{\sffamily\bfseries\upshape (iv)}, $R(s)= \Rsim(s)$. 
Moreover, $s\in \Rsim(x)$ implies $\Rsim(s)  \subseteq \Rsim(\Rsim(x))=\Rsim(x)$, and since, 
by \textcolor{lipicsGray}{\sffamily\bfseries\upshape (ii)}, $\Rsim(x) \subseteq R(x)$, we obtain $\Rsim(s)\subseteq R(x)$. 
Thus,  $R(s)\subseteq R(x)$ holds, so that  $R(s) \cap \sigma \neq \varnothing$ entails
$R(x) \cap \sigma \neq \varnothing$, and in turn, by \textcolor{lipicsGray}{\sffamily\bfseries\upshape (iii)}, $R(x)=\Rsim(x)$.

$(\supseteq)$ Let $x\in \Sigma$ such that  $R(x) \cap \sigma \neq \varnothing$. By \textcolor{lipicsGray}{\sffamily\bfseries\upshape (iv)}, $R(x)=\Rsim(x)$, so that $\Rsim(x)\cap \sigma \neq \varnothing$. By \textcolor{lipicsGray}{\sffamily\bfseries\upshape (v)}, $\Rsim(x)\cap \post^*(I) \neq \varnothing$.

\medskip
\noindent
Let us now prove \eqref{corr-stm-expl_sound2}.
Let $x \in \Sigma$ such that $\Psim(x)\cap \post^*(I) \neq \varnothing$.
By definition of \(P\), we have that for all $x\in \Sigma$, 
$P(x)=\{y \in \Sigma \mid R(x) = R(y)\}$. 
Since $\Psim(x)\subseteq \Rsim(x)$, we have that $\Rsim(x)\cap \post^*(I) \neq \varnothing$.
Let $y\in \Psim(x)$. Observe that 
$\Psim(x) = \{z\in \Sigma \mid \Rsim(x)=\Rsim(z)\}$.  
Thus, $\Rsim(x)=\Rsim(y)$, so that 
$\Rsim(y)\cap \post^*(I) \neq \varnothing$. Thus, by \eqref{corr-stm-expl_sound}, $\Rsim(y) =  R(y)$ and 
$R(y) \cap \sigma \neq \varnothing$.
In particular, $R(x) = \Rsim(x) = \Rsim(y) = R(y)$.
Therefore, $\Psim(x) \subseteq \{y \in \Sigma \mid R(x) = R(y)\}=P(x)$.
On the other hand, if $z \in P(x)$ then $R(z) = R(x)$,
so that $R(x) \cap \sigma \neq \varnothing$ implies 
$R(z) \cap \sigma \neq \varnothing$. 
Thus, by \eqref{corr-stm-expl_sound}, $\Rsim(z) = R(z)$. Hence, 
$\Rsim(z)= R(z) = R(x)= \Rsim(x)$, 
thus proving that $P(x) \subseteq \Psim(x)$, and therefore $P(x) = \Psim(x)$ holds.
Moreover, $R(x)\cap \sigma \neq \varnothing$, thus proving~\eqref{corr-stm-expl_sound2}.
\end{proof}

\terminationExplicit*
\begin{proof}
We first observe that each \emph{Search} iteration adds some new state to \(\sigma\) through the update at line~\ref{loc:ex_sigmaupdate}. Thus, since \(\Sigma\) has finitely many elements, Algorithm~\ref{algo:explicit_sound} will always execute a finite number of \emph{Search} iterations.
Similarly, executing the \emph{Refine} block is guaranteed to remove some state from at least one principal of the current relation, and since each principal has an initial finite number of elements, the overall number of \emph{Refine} iterations is also finite.
Therefore, since every iteration of the while-loop either executes a \emph{Search} or a \emph{Refine}, the algorithm terminates after a finite number of iterations.
Finally, we observe that each iteration computes a finite number of operations, so that this completes the proof.
\end{proof}

\correctnessNonterminatingExplicit*
\begin{proof}
We first observe that at termination \(V = \varnothing\) and one of the two following holds:
\begin{align}
U &= \varnothing \\
U &\neq \varnothing \land (\post(\sigma) \subseteq \sigma) 
\end{align}
Corresponding to, respectively, the case where the algorithm executes line \ref{loc:ntex_algo2call} or \ref{loc:ntex_return2}.
Moreover, in the case $U \neq \varnothing$, we observe that \(\post(\sigma) \subseteq \sigma\), 
together with \(I \subseteq \sigma\) allows us to conclude that \(\post^*(I) \subseteq \sigma\) by fixpoint definition of \(\post^*(I)\), 
and Inv\textsubscript{2} implies \(\sigma = \post^*(I)\).
Moreover, by correctness of Algorithm \ref{algo6_refine_ex} we get that, at line~\ref{loc:ntex_algo2call}, \(\sigma \subseteq \sigma' \subseteq \post^*(I)\) and thus \(\sigma = \sigma' = \post^*(I)\).

We now show that the following facts on the  output pair $\tuple{R,\sigma}$  of Algorithm~\ref{algo:nonterminating_ex}  hold:
\begin{romanenumerate}
\item  For all $x\in \Sigma$, $\Rsim(x) \subseteq R(x)$.\\ 
This holds because the \emph{Refine} block of Algorithm~\ref{algo:nonterminating_ex} is always correct, so that,  for all $x\in \Sigma$, 
 $\Rsim(x)\subseteq R(x) \subseteq R_i(x)$ holds for every iteration of Algorithm~\ref{algo:nonterminating_ex}.
Moreover, correctness of Algorithm~\ref{algo6_refine_ex} ensures that line \ref{loc:ntex_algo2call} preserves this property.

\item $x\in \post^*(I) \Ra \exists s \in \sigma \ldotp R(x) = R(s)$.\\
We distinguish the two possible cases at termination:

\noindent \((U=\varnothing)\): We proceed by induction on $n\in \bN$ such
that $x\in \post^n(I)$.
For $n=0$, $x\in I$, so that by initialization of \(\sigma\), \(x \in \sigma\).
For the inductive case we proceed by induction on $n\in \bN$ such
that $x\in \post^n(I)$:
If $x\in \post^{n+1}(I)$ then there exists $y\in \post^n(I)$ such that 
$y \sra x$, so that, by inductive hypothesis, $R(y) = R(s')$ for some $s'\in\sigma$. 
Since $y\sra x$, $V=\varnothing$ imply
$R(y)\subseteq \pre(R(x))$. Thus, since $R$ is reflexive, $s' \in R(s') = R(y)$ holds, 
implying $s' \in \pre(R(x))$, i.e., 
$R(x)\cap \post(\sigma)\neq \varnothing$. Then, $U=\varnothing$ implies $\exists s \in \sigma \ldotp R(s) = R(x)$.

\noindent \((U\neq\varnothing)\): Since \(\sigma = \post^*(I)\), 
then \(x \in \post^*(I) = \sigma\) and the property trivially holds.

\item $(\exists s \in \sigma \ldotp R(s) = R(x)) \Ra R(x) = \Rsim(x)$.\\ 
We distinguish the two possible cases at termination:

\noindent \((U\neq\varnothing)\): 
By reflexivity, \(s \in R(s) = R(x)\) implies \(R(x) \cap \sigma\neq\varnothing\), and
thus correctness of Algorithm~\ref{algo6_refine_ex} let us conclude \(R(x) = \Rsim(x)\).

\noindent \((U=\varnothing)\): Let \(\Sim\) be the basic simulation algorithm taking in input 
a relation $R$ and computing, for finite state systems, $\Rsim$.

By (i), $\Rsim(z) \subseteq R(z) \subseteq R_i(z)$, for all $z\in \Sigma$. Moreover, \(\Sim(R) = \Rsim\) because \(\Rsim =
\Sim(\Rsim) \subseteq \Sim(R) \subseteq \Sim(R_i) = \Rsim\). 
Let us assume, by contradiction, that
\(\bS \ud \{R(x), \mid x\in \Sigma, R(x) \neq \Rsim(x), \exists s \in \sigma \ldotp R(x) = R(s)\}\neq \varnothing\), so that
every principal in \(\bS\) will be refined at some iteration of 
\(\Sim\) on input \(R\). 
Let \(R(x) \in \bS\) be the first principal in \(\bS\) to be refined by $\Sim$
at some iteration $\mathit{it}$ whose current relation is $R'$ with \(R'\preceq R\), so that
each principal in \(\bS\) is  a principal for \(R'\) since no principal in $\bS$ was refined before this iteration $\mathit{it}$.
Let $R'(y)\subseteq R(y)$ be the principal of $R'$ used in this iteration $\mathit{it}$ 
to refine \(R(x)\), so that \(x \sraa{a} y\),
\(R'(x) = R(x)\not\subseteq \pre_a(R'(y))\). 
We have that \(x \sraa{a} y$, along with 
\( \exists s \in \sigma \ldotp R(s) = R(x) \) and \(V=\varnothing\) entails 
\(R(x) \subseteq \pre_a(R(y))\). 
Hence, by reflexivity \(s \in R(s) = R(x) \subseteq \pre_a(R(y))\) 
implies \(R(y) \cap \post(\sigma) \neq \varnothing\), 
so that \(U = \varnothing\) entails \(\exists s' \in \sigma \ldotp R(y) = R(s')\).

Now, either \(R(y) \not\in \bS\), which implies \(R(y) = \Rsim(y)\) and therefore \(R(y) = R'(y)\),
or \(R(y) \in \bS\) and since no principal in \(\bS\) was refined before, \(R(y) = R'(y)\).

This lets us conclude that \(R(x) \not\subseteq \pre_a(R'(y)) = \pre_a(R(y))\),
which is a contradiction to \(R(x) \subseteq \pre_a(R(y))\).
Thus \(\bS = \varnothing\) and, in turn \(R(x) = \Rsim(x)\).

\item $(\exists s \in \sigma \ldotp R(s) = R(x)) \Ra P(x)\cap\post^*(I) = \Psim(x)\cap\post^*(I)$.\\
We first show that
\begin{equation}\label{eq:ntex_blockinclusion}
(\exists s \in \sigma \ldotp R(s) = R(x)) \Ra P(x) \subseteq \Psim(x).
\end{equation}
Consider an element \(y \in P(x)\), then by definition of \(P\) it holds \(R(y) = R(x)\) and thus \(R(y) = R(s)\),
moreover, by \textcolor{lipicsGray}{\sffamily\bfseries\upshape (iii)} we get \(\Rsim(y) = R(y) = R(x) = \Rsim(x)\) and thus by definition of \(\Psim\) it follows that \(y \in \Psim(x)\),
meaning \(P(x) \subseteq \Psim(x)\), which proves \eqref{eq:ntex_blockinclusion}. The inclusion \(P(x)\cap\post^*(I) \subseteq \Psim(x)\cap\post^*(I)\) follows by definition from \eqref{eq:ntex_blockinclusion}.
On the other hand, pick some \(y \in \Psim(x)\cap\post^*(I)\), then by definition of \(\Psim\) we get \(\Rsim(y) = \Rsim(x)\),
by \textcolor{lipicsGray}{\sffamily\bfseries\upshape (ii)}, we get that \(y \in \post^*(I)\) entails \(R(y) = R(s')\) for some \(s' \in \sigma\)
and \textcolor{lipicsGray}{\sffamily\bfseries\upshape (iii)} implies both \(R(y) = \Rsim(y) = \Rsim(x) = R(x)\).
Thus by definition of \(P\), \(y \in P(x)\) and therefore \(\Psim(x)\cap\post^*(I) \subseteq P(x)\cap\post^*(I)\).

\item \(s \in \sigma \Ra \Psim(s) = \Psim(P(s))\):
Take a state \(s \in \sigma\), then it holds that, by~\eqref{eq:ntex_blockinclusion}, \(s \in P(s) \subseteq \Psim(s)\), 
and by definition of additive lifting it holds \(\bigcup_{x \in \{s\}}\Psim(x) \subseteq \bigcup_{x \in P(s)}\Psim(x) \subseteq \bigcup_{x \in \Psim(s)}\Psim(x)\) 
that, in turn, entails \(\Psim(s) \subseteq \Psim(P(s)) \subseteq \Psim(\Psim(s)) = \Psim(s)\), 
and therefore \(\Psim(s) = \Psim(P(s))\).

\item $\sigma \subseteq \post^*(I)$.\\ 
This holds since \(\sigma_i \subseteq \post^*(I)\) and updates at lines~\ref{loc:ntex_sigmaupdate},~\ref{loc:ntex_sigma_expand_update} preserve this property.

\end{romanenumerate}

Let us now show \eqref{corr-stm-ntex}: 
Consider \(x \in \post^*(I)\), then by \textcolor{lipicsGray}{\sffamily\bfseries\upshape (ii)} there exists some \(s \in \sigma\) s.t. \(R(s) = R(x)\)
and by \textcolor{lipicsGray}{\sffamily\bfseries\upshape (iii)} \(R(x) = \Rsim(x)\), proving that \(\Rsim(x) \in \{R(y) \mid y \in \sigma\}\).
On the other hand, consider a state \(s \in \sigma\), then by \textcolor{lipicsGray}{\sffamily\bfseries\upshape (vi)} we have \(s \in \post^*(I)\),
moreover, by \textcolor{lipicsGray}{\sffamily\bfseries\upshape (iii)} we get \(R(s) = \Rsim(s)\) and thus we get \(R(s) \in \{\Rsim(y) \mid y \in \post^*(I)\}\).

We now prove \eqref{corr-stm-ntex2}: 
Consider a block \(\Psim(z)\) s.t. \(\Psim(z) \cap \post^*(I) \neq \varnothing\), then \(\Psim(z) = \Psim(x)\) for some \(x \in \post^*(I)\).
By \textcolor{lipicsGray}{\sffamily\bfseries\upshape (ii)}, \(R(x) = R(s)\) for some \(s \in \sigma\) and by \textcolor{lipicsGray}{\sffamily\bfseries\upshape (iii)} \(\Rsim(x) = R(x) = R(s) = \Rsim(s)\) meaning \(\Psim(s) = \Psim(x) = \Psim(z)\), moreover, by \textcolor{lipicsGray}{\sffamily\bfseries\upshape (iv)} we get \(\Psim(s) \cap \post^*(I) = P(s) \cap\post^*(I)\) 
and by reflexivity of \(P\) we get \(\Psim(z) \cap \post^*(I) \in \{B \cap \post^*(I) \mid B \in P,\, B \cap \sigma \neq \varnothing\}\).

For the other inclusion, we consider \(P(z)\) s.t. \(P(z) \cap \sigma \neq \varnothing\) 
and thus \(P(z) = P(s)\) for some \(s \in \sigma\) and \(R(z) = R(s)\), by definition of \(P\).
By \textcolor{lipicsGray}{\sffamily\bfseries\upshape (iv)}, \(P(z) \cap \post^*(I) = \Psim(z) \cap \post^*(I)\),
moreover, \textcolor{lipicsGray}{\sffamily\bfseries\upshape (v)} entails \(s \in P(z) \cap \post^*(I)\) and thus \(s \in \Psim(z) \cap \post^*(I)\),
meaning \(\Psim(z) \cap \post^*(I) \neq \varnothing\) and thus proving 
\(P(z) \cap \post^*(I) \in \{B \cap \post^*(I) \mid B \in \Psim,\, B \cap \post^*(I) \neq \varnothing\}\).

Finally, we prove the two inclusions of \eqref{corr-stm-ntex3}:

(\(\subseteq\)): Pick a block \(\Psim(x)\) s.t. \(\Psim(x) \cap \post^*(I) \neq \varnothing\) 
and consider \(z \in \Psim(x)\cap\post^*(I)\), then by definition of \(\Psim\) we have \(\Rsim(x) = \Rsim(z)\),
moreover by \textcolor{lipicsGray}{\sffamily\bfseries\upshape (ii)} there exists some state \(s \in \sigma\) s.t. \(R(s) = R(z)\) and by \textcolor{lipicsGray}{\sffamily\bfseries\upshape (iii)} 
we get \(\Rsim(z) = R(z) = R(s) = \Rsim(s)\) therefore proving \(s \in \Psim(z) = \Psim(x)\).
Moreover, by definition of \(\Psim\) it holds \(\Psim(x) = \Psim(s)\) and by \textcolor{lipicsGray}{\sffamily\bfseries\upshape (v)} we get \(\Psim(x) = \Psim(P(s))\) thus proving the inclusion since \(s \in P(s)\).

(\(\supseteq\)): We consider a state \(s \in \sigma\) and the corresponding \(P(s)\), then by \textcolor{lipicsGray}{\sffamily\bfseries\upshape (v)} we get \(\Psim(P(s)) = \Psim(s) \in \Psim\) and by \textcolor{lipicsGray}{\sffamily\bfseries\upshape (vi)} we get that, since \(s \in \Psim(s)\), then \(\Psim(s) \cap \post^*(I) \neq \varnothing\), proving the inclusion.
\end{proof}

\terminationNonterminatingExplicit*
\begin{proof}
We first observe that executing the \emph{Refine} block is guaranteed to remove some state from at least one principal of the current relation, and since each principal has an initial finite number of elements, the overall number of \emph{Refine} iterations is finite.
Then, every \emph{Expand} iteration either adds some new state to \(\sigma\) through the update at line~\ref{loc:ntex_sigma_expand_update}, or it reaches the return statement at line~\ref{loc:ntex_algo2call}, and thus the number of \emph{Expand} iterations is also finite.
Moreover, every \emph{Search} iteration will either add some new state to \(\sigma\) through the update at line~\ref{loc:ntex_sigmaupdate}, or some new state to $\Ubad$ through the update at line~\ref{loc:ntex_ubadupdate}. Thus, since \(\Sigma\) has finitely many elements,  
the algorithm will always execute a finite number of \emph{Search} iterations which expand \(\sigma\).
Finally, the number of \emph{Search} iterations expanding \(\Ubad\) occurring in between two iterations resetting \(\Ubad\) to \(\varnothing\) is also finite, since \(\Ubad \subseteq \Sigma\) and \(\Sigma\) is finite, and observing that the iterations executing \(\Ubad := \varnothing\) are, in turn, either \emph{Refine} iterations, non terminating \emph{Expand} iterations, or \emph{Search} iterations expanding \(\sigma\), which we have shown to be finitely many allows us to conclude that the total number of \emph{Search} iterations is also finite.

Therefore, since every iteration of the while-loop executes either a \emph{Search}, \emph{Refine} or \emph{Expand} iteration, the algorithm terminates after a finite number of iterations.
Observing that each iteration computes a finite number  of operations, and in particular the execution of Algorithm~\ref{algo6_refine_ex} at line~\ref{loc:ntex_algo2call} is guaranteed to terminate by Theorem~\ref{thm:algo_sound_finite_termination} completes the proof.
\end{proof}

\twoprQO*
\begin{proof}
We first observe that for any preorder \(R\), \(P_R = R \cap R^{-1}\) holds.
In fact, for all \(x \in \Sigma\), assume \(y \in P_R(x)\), then \(R(x) = R(y)\) and by reflexivity of \(R\), \(x R y R x\) holds, thus proving \((x, y), (y, x) \in R\cap R^{-1}\).
Conversely, assume \((x, y), (y, x) \in R\cap R^{-1}\) and consider \(z \in R(x)\), thus \(y R x R z\) an by transitivity \(z \in R(x)\), proving \(R(x) \subseteq R(y)\). The proof for \(R(y) \subseteq R(x)\) is symmetric.

Finally, we observe that since \(R\) is a preorder, then \(R^{-1}\) is a preorder too and therefore by the previous result \(Q_R = R^{-1} \cap (R^{-1})^{-1}\).
Observing that \(R^{-1} \cap (R^{-1})^{-1} = R \cap R^{-1}\) completes the proof.
\end{proof}

\twoprExact*
\begin{proof}
We first observe that \(\forall x \in \Sigma \ldotp R(x) = R(P_R(x))\), by additive lifting of \(R\) over the elements of \(P_R(x)\), and the fact that \(z \in P_R(x) \Lra R(z) = R(x)\). 
Moreover, \(P_R\) and \(Q_R\) are both partitions of \(\Sigma\), thus for all \(x\in \Sigma\), both \(x \in P_R(x)\) and \(x \in Q_R(y)\).
We prove the two inclusions for the property:\\
(\(\subseteq\)): Let \(y \in R(x)\), then by definition of \(P_R\), \(x \in P_R(x)\) and similarly \(y \in Q_R(y)\).
We show that \(Q_R(y) \in \tau_R(P_R(x))\), in fact \(\forall z \in Q_R(y) \ldotp x \in R^{-1}(y) = R^{-1}(z)\) holds,
and since \(R(x) = R(P_R(x))\) then \(Q_R(y) \subseteq R(P_R(x))\).
Therefore, \(Q_R(y) \in \tau_R(P_R(x))\) hols and as a consequence \(y \in Q_R(y) \subseteq \cup\tau_R(P_R(x)) = R_{\tuple{P_R, \tau_R, Q_R}}(x)\).

(\(\supseteq\)): Let \(y \in R_{\tuple{P_R, \tau_R, Q_r}}(x)\), then \(Q_R(y) \in \tau_R(P_R(x))\) and thus \(Q_R(y) \subseteq R(P_R(x))\) holds.
Since \(Q_R\) is a partition and \(R(P_R(x)) = R(x)\), then we get \(y \in Q_R(y) \subseteq R(x)\), completing the proof.
\end{proof}

\twoprReflexiveCharacterization*
\begin{proof}
We prove the two implications:

(\(\Ra\)): Let \(B \in P\) and \(x \in B\), then by reflexivity \(x \in R_{\tuple{P, \tau, Q}}(x) = \cup\tau(P(x)) = \cup\tau(B)\).

(\(\La\)): Let \(x \in \Sigma\), then by extensivity of \(\cup\tau\) it holds \(x \in P(x) \subseteq \cup\tau(P(x)) = R_{\tuple{P, \tau, Q}}(x)\).
\end{proof}

\correctnessSymbolic*
\begin{proof}
We show some preliminary properties for the algorithm:
\begin{romanenumerate}
\item At each iteration: \(B \subseteq \cup\tau(B)\) for all \(B \in P\).\\
This holds at initialization because $R_i$ is a preorder and we proceed to show that the property is preserved at each \emph{Refine} step.
Suppose that a triple \(\tuple{a, B, C}\) is picked for refinement, and let \(\tuple{P', \tau', Q'}\) be the 2PR triple after execution of the \emph{Refine} block.
We notice that \(\forall X \in P \ldotp X \not\subseteq B \Ra \cup\tau(X) = \cup\tau'(X)\), and that \(\cup\tau(P(B'')) = \cup\tau'(B'')\) since the only instruction removing states from \(\cup\tau(X)\), for some block \(X\) is line~\ref{loc:sym_refine_principal}.
We now take \(x \in B' = B \cap \pre_a(C)\) and proceed to show that \(x \in \cup\tau'(B')\).
Since \(x \in \pre_a(C) \subseteq \pre_a(\cup\tau(C))\) then \(x \in S\) and thus \(Q(x) \cap S \neq \varnothing\). 
We now distinguish two cases in order to show that \(Q'(x) \subseteq S\):

(\(Q(x) \subseteq S\)): Thus \(Q(x)\) is not split and \(Q'(x) = Q(x) \subseteq S\).

(\(Q(x) \not\subseteq S\)): Thus \(Q(x)\) is split and \(Q'(x) = X \cap S \subseteq S\) by definition.

Therefore, since \(Q'(x) \subseteq S\) holds, then we get \(Q'(x) \in \tau(B')\) after line~\ref{loc:sym_refine_principal} and thus we conclude \(x \in \cup\tau'(B')\), completing the proof.

\item At each iteration: $\forall x\in \Sigma \ldotp \Rsim(x) \subseteq R_{\tuple{P, \tau, Q}}(x) \subseteq R_i(x)$.\\ 
This holds at initialization since \(R_{\tuple{P_{R_i}, \tau_{R_i}, Q_{R_i}}} = R_i\), and we proceed to show that every \emph{Refine} step preserves the invariant.
Let us consider a \emph{Refine} iteration, the corresponding 2PR triple \(\tuple{P, \tau, Q}\) for which the property holds, the triple \(\tuple{a, B, C} \in V\) selected for refinement and the 2PR triple \(\tuple{P', \tau', Q'}\) at the end of the \emph{Refine} block.
We first note that \(x \not\in B \Ra \cup\tau(P(x)) = \cup\tau'(P'(x))\) and  similarly \(x \in B \smallsetminus \pre_a(C) \Ra \cup\tau(P(x)) = \cup\tau'(P'(x))\) since the \emph{Refine} step splits no block in \(P\) other than \(B\) and updates the underlying principal for states in \(B \cap \pre_a(C)\) only.
Thus we proceed by assuming \(x \in B \cap \pre_a(C)\) and show that for every \(y \in \cup\tau(P(x)) \smallsetminus \cup\tau'(P'(x))\) it holds that \(y \not\in \Rsim(x)\).

In order to do so we observe that \(Q'(y) \not\in \tau'(P'(x))\) implies that line~\ref{loc:sym_refine_principal} removed the block \(Q'(y)\), thus meaning \(Q'(y) \not\subseteq S\).
We now show that \(y \not\in S\) by distinguishing two cases:

(\(Q(y) \cap S = \varnothing\)):
This case entails \(y \not\in S\) trivially, by definition of intersection.

(\(Q(y) \cap S \neq \varnothing\)):
Since \(Q'(y) \subseteq Q(y)\) then \(Q(y) \not\subseteq S\), meaning the block \(Q(y)\) was split during the \emph{Refine} step, meaning that either \(Q'(y) = Q(y) \cap S\) or \(Q'(y) = Q(y) \smallsetminus S\).
We note that, the case \(Q'(y) = Q(y) \cap S\) leads to a contradiction since \(Q'(y) \not\subseteq S\) and thus we infer \(Q'(y) = Q(y) \smallsetminus S\).
Observing that \(y \in Q'(y)\) lets us conclude \(y \not\in S\).

We now observe that, since \(x \in B \cap \pre_a(C)\), then there exists some \(x' \in C\) s.t. \(x \sraa{a} x'\).
Since \(y \not\in S\) and \(\Rsim(x') \subseteq \cup\tau(P(x'))\) we can conclude that \(y \not\in \pre_a(\Rsim(x')) \subseteq S\), which together with \(x \sra x'\) implies \(y \not\in \Rsim(x)\).

\item At termination: $x\in \post^*(I)$ $\Ra$ $R_{\tuple{P, \tau, Q}}(x)\cap \sigma \neq \varnothing$.\\
This is proven by induction on $n\in \bN$ such that $I\sra^n x$. 
If $n=0$ then $x\in I$, so that, since $P(x) \subseteq \cup\tau(P(x))$ following \textcolor{lipicsGray}{\sffamily\bfseries\upshape (i)}, and \(x \in P(x)\), then $\cup\tau(P(x))\cap (I\cup \post(\sigma))\neq \varnothing$. 
Hence, $U=\varnothing$ implies that  $\cup\tau(P(x))\cap \sigma \neq \varnothing$. 
If $n>0$ then $I\sra^n x' \sraa{a} x$, and
by inductive hypothesis, $\cup\tau(P(x'))\cap \sigma \neq \varnothing$. 
Therefore, since $V=\varnothing$, and \(P(x') \cap \pre_a(P(x)) \neq \varnothing\), then $\cup\tau(P(x')) \subseteq \pre_a(\cup\tau(P(x)))$ must hold. 
Thus, $\pre_a(\cup\tau(P(x))) \cap \sigma \neq \varnothing$, so that $\cup\tau(P(x)) \cap \post_a(\sigma) \neq \varnothing$. 
Hence, $U=\varnothing$ implies  $\cup\tau(P(x))\cap \sigma \neq \varnothing$.

\item At termination: $\forall x\in \Sigma \ldotp R_{\tuple{P, \tau, Q}}(x)\cap \sigma \neq \varnothing \Ra R_{\tuple{P, \tau, Q}}(x) = \Rsim(x)$.\\
Let $\Sim$ be the basic simulation algorithm as recalled in Algorithm~\ref{basicsim}.
By \textcolor{lipicsGray}{\sffamily\bfseries\upshape (ii)}, $\Rsim \subseteq R_{\tuple{P, \tau, Q}} \subseteq R_i$.
Thus, $\Sim(R_{\tuple{P, \tau, Q}}) = \Rsim$ because
$\Rsim = \Sim(\Rsim) \subseteq \Sim(R_{\tuple{P, \tau, Q}}) \subseteq \Sim(R_i) = \Rsim$.
Assume, by contradiction, that $\bS \ud \{R_{\tuple{P, \tau, Q}}(x) \in \wp(\Sigma) \mid R_{\tuple{P, \tau, Q}}(x)\cap \sigma \neq \varnothing, R_{\tuple{P, \tau, Q}}(x) \neq \Rsim(x)\}\neq \varnothing$,
so that each $R_{\tuple{P, \tau, Q}}(x)\in \bS$, will be refined 
at some iteration of $\Sim(R_{\tuple{P, \tau, Q}})$.
Then, let \(R_{\tuple{P, \tau, Q}}(x) \in \bS\) be the first principal in \(\bS\) such that  $R_{\tuple{P, \tau, Q}}(x)$  is refined by $\Sim$
at some iteration $\mathit{it}$ whose current relation is $R'$. 
Thus, each principal \(R_{\tuple{P, \tau, Q}}(z)\in\bS\) is such that $R_{\tuple{P, \tau, Q}}(z)=R'(z)$ because no principal 
in $\bS$ was refined  by $\Sim$ before this iteration $\mathit{it}$.
Let $R'(y)\subseteq R_{\tuple{P, \tau, Q}}(y)$ be the principal of $R'$ used by $\Sim$ in this iteration $\mathit{it}$ 
to refine \(R_{\tuple{P, \tau, Q}}(x)\), so that \(x \sraa{a} y\),  
\(R'(x)=R_{\tuple{P, \tau, Q}}(x) \not\subseteq \pre_a(R'(y))\).
By termination of Algorithm~\ref{algo:symbolic_sound} it holds that $R_{\tuple{P, \tau, Q}}(x)\cap \sigma \neq \varnothing$, \(P(x) \cap \pre(P(y)) \neq \varnothing\) and $V=\varnothing$ entail $\cup\tau(P(x))\subseteq \pre_a(\cup\tau(P(y)))$. 
Hence, $\cup\tau(P(x))\cap \sigma \neq \varnothing$ implies  
$\cup\tau(P(y))\cap \post_a(\sigma)\neq \varnothing$, 
so that $U=\varnothing$ entails $\cup\tau(P(y))\cap \sigma\neq \varnothing$, that is \(R_{\tuple{P, \tau, Q}}(y) \cap \sigma\neq \varnothing\).
If $R_{\tuple{P, \tau, Q}}(y)\not\in\bS$ then $R_{\tuple{P, \tau, Q}}(y)\cap \sigma \neq \varnothing$ implies $R_{\tuple{P, \tau, Q}}(y)=\Rsim(y)\subseteq
R'(y)\subseteq R_{\tuple{P, \tau, Q}}(y)$, so that $R'(y)= R_{\tuple{P, \tau, Q}}(y)$ holds. 
If $R_{\tuple{P, \tau, Q}}(y)\in\bS$ then, since $R_{\tuple{P, \tau, Q}}(x)$ is the first principal in $\bS$ to be refined by $\Sim$, we 
have that $R'(y)= R_{\tuple{P, \tau, Q}}(y)$ holds. Thus, in both cases, $R'(y)= R_{\tuple{P, \tau, Q}}(y)$ must hold, 
so that $R_{\tuple{P, \tau, Q}}(x) \not\subseteq \pre_a(R_{\tuple{P, \tau, Q}}(y))$, which gives a contradiction to \(\cup\tau(P(x))\subseteq \pre_a(\cup\tau(P(y)))\).
 Thus, $\bS=\varnothing$, and, in turn, $R_{\tuple{P, \tau, Q}}(x) = \Rsim(x)$.
 
 \item At each iteration: $\sigma\subseteq \post^*(I)$. \\
This follows by \(\sigma_i \subseteq \post^*(I)\) and because $\sigma$ is updated at line~\ref{loc:sym_sigmaupdate} of Algorithm~\ref{algo:symbolic_sound} by 
adding $s\in I\cup \post(\sigma)$ and $\post^*(I)$ is the least fixpoint
of $\lambda X.I\cup \post(X)$. 

\item At termination: \(P_{\tuple{P, \tau, Q}}(x) \cap \sigma \neq \varnothing \Ra P_{\tuple{P, \tau, Q}}(x) = \Psim(x)\).\\
We first observe that, assuming \(P_{\tuple{P, \tau, Q}}(x) \cap \sigma \neq \varnothing\), and considering \(s \in P_{\tuple{P, \tau, Q}}(x) \cap \sigma\) we find that
\(R_{\tuple{P, \tau, Q}}(x) = \cup\tau(P(x)) = \cup\tau(P(s)) = R_{\tuple{P, \tau, Q}}(s)\).
Moreover, since by \textcolor{lipicsGray}{\sffamily\bfseries\upshape (i)}, \(s \in P(s) \subseteq \cup\tau(P(s))\) then \(R_{\tuple{P, \tau, Q}}(s) \cap \sigma \neq \varnothing\).
Therefore, by \textcolor{lipicsGray}{\sffamily\bfseries\upshape (iv)}, \(\Rsim(x) = R_{\tuple{P, \tau, Q}}(x) = R_{\tuple{P, \tau, Q}}(s) = \Rsim(s)\).
We prove the two inclusions of \(P_{\tuple{P, \tau, Q}}(x) = \Psim(x)\) separately:

(\(\subseteq\)):
Assume \(y \in P_{\tuple{P, \tau, Q}}(x)\), then \(\cup\tau(P(y)) = \cup\tau(P(x))\) and since \(\cup\tau(P(x)) \cap \sigma \neq \varnothing\) then \(\cup\tau(P(y)) \cap \sigma \neq \varnothing\).
Finally, \textcolor{lipicsGray}{\sffamily\bfseries\upshape (iv)} entails \(\Rsim(y) = \cup\tau(P(y)) = \cup\tau(P(x)) = \Rsim(x)\) thus proving \(y \in \Psim(x)\).

(\(\supseteq\)):
Assume \(y \in \Psim(x)\), then \(\Rsim(y) = \Rsim(x) = \Rsim(s)\).
Moreover, since \(s \in \sigma\), and \(s \in \Rsim(s)\), we get \(\Rsim(y) \cap \sigma \neq \varnothing\) and since, by \textcolor{lipicsGray}{\sffamily\bfseries\upshape (ii)}, \(\Rsim(y) \subseteq R_{\tuple{P, \tau, Q}}(y)\) then \(R_{\tuple{P, \tau, Q}}(y) \cap \sigma \neq \varnothing\).
Therefore, \textcolor{lipicsGray}{\sffamily\bfseries\upshape (iv)} entails \(R_{\tuple{P, \tau, Q}}(y) = \Rsim(y)\), 
and thus \(R_{\tuple{P, \tau, Q}}(y) = \Rsim(y) = \Rsim(x) = R_{\tuple{P, \tau, Q}}(x)\) holds,
proving \(\cup\tau(P(y)) = \cup\tau(P(x))\) and thus \(y \in P_{\tuple{P, \tau, Q}}(x)\).

\end{romanenumerate}
\noindent
Let us now show the equality \eqref{corr-stm-sym_sound}. 

$(\subseteq)$
Let $x\in \Sigma$ such that  $\Rsim(x)\cap \post^*(I)\neq \varnothing$. Thus, there exists $s\in \post^*(I)$
such that $s\in \Rsim(x)$. By \textcolor{lipicsGray}{\sffamily\bfseries\upshape (iii)}, $R_{\tuple{P, \tau, Q}}(s) \cap \sigma \neq \varnothing$. By \textcolor{lipicsGray}{\sffamily\bfseries\upshape (iv)}, $R_{\tuple{P, \tau, Q}}(s)= \Rsim(s)$. 
Moreover, $s\in \Rsim(x)$ implies $\Rsim(s) \subseteq \Rsim(\Rsim(x))=\Rsim(x)$, and since, 
by \textcolor{lipicsGray}{\sffamily\bfseries\upshape (ii)}, $\Rsim(x) \subseteq R_{\tuple{P, \tau, Q}}(x)$, we obtain $\Rsim(s)\subseteq R_{\tuple{P, \tau, Q}}(x)$. 
Thus,  $R_{\tuple{P, \tau, Q}}(s)\subseteq R_{\tuple{P, \tau, Q}}(x)$ holds, so that  $R_{\tuple{P, \tau, Q}}(s) \cap \sigma \neq \varnothing$ entails
$R_{\tuple{P, \tau, Q}}(x) \cap \sigma \neq \varnothing$, and in turn, by \textcolor{lipicsGray}{\sffamily\bfseries\upshape (iv)}, $R_{\tuple{P, \tau, Q}}(x)=\Rsim(x)$.

$(\supseteq)$ Let $x\in \Sigma$ such that  $R_{\tuple{P, \tau, Q}}(x) \cap \sigma \neq \varnothing$. By \textcolor{lipicsGray}{\sffamily\bfseries\upshape (iv)}, $R_{\tuple{P, \tau, Q}}(x)=\Rsim(x)$, so that $\Rsim(x)\cap \sigma \neq \varnothing$. By \textcolor{lipicsGray}{\sffamily\bfseries\upshape (v)}, $\Rsim(x)\cap \post^*(I) \neq \varnothing$.

\medskip
\noindent
Let us now prove \eqref{corr-stm-sym_sound2}.
Let $x \in \Sigma$ such that $\Psim(x)\cap \post^*(I) \neq \varnothing$.
Since $\Psim(x)\subseteq \Rsim(x)$, we have that $\Rsim(x)\cap \post^*(I) \neq \varnothing$.
Let $y\in \Psim(x)$. Observe that 
$\Psim(x) = \{z\in \Sigma \mid \Rsim(x)=\Rsim(z)\}$.  
Thus, $\Rsim(x)=\Rsim(y)$, so that $\Rsim(y)\cap \post^*(I) \neq \varnothing$. 
Thus, by \eqref{corr-stm-sym_sound}, $\Rsim(y) =  R_{\tuple{P, \tau, Q}}(y)$ 
and $R_{\tuple{P, \tau, Q}}(y) \cap \sigma \neq \varnothing$.
In particular, $R_{\tuple{P, \tau, Q}}(x) = \Rsim(x) = \Rsim(y) = R_{\tuple{P, \tau, Q}}(y)$.

Therefore, $\Psim(x) \subseteq \{y \in \Sigma \mid \cup\tau(P(x)) = \cup\tau(P(y))\}=P_{\tuple{P, \tau, Q}}(x)$.
On the other hand, if $z \in P_{\tuple{P, \tau, Q}}(x)$ then $R_{\tuple{P, \tau, Q}}(z) = \cup\tau(P(z)) = \cup\tau(P(x)) = R_{\tuple{P, \tau, Q}}(x)$,
so that $\cup\tau(P(x)) \cap \sigma \neq \varnothing$ implies 
$R_{\tuple{P, \tau, Q}}(z) \cap \sigma \neq \varnothing$. 
Thus, by \eqref{corr-stm-sym_sound}, $\Rsim(z) = R_{\tuple{P, \tau, Q}}(z)$. Hence, 
$\Rsim(z)= R_{\tuple{P, \tau, Q}}(z) = R_{\tuple{P, \tau, Q}}(x)= \Rsim(x)$, 
thus proving that $P_{\tuple{P, \tau, Q}}(x) \subseteq \Psim(x)$, and therefore $P_{\tuple{P, \tau, Q}}(x) = \Psim(x)$ holds.
Finally, since $R_{\tuple{P, \tau, Q}}(x)\cap \sigma \neq \varnothing$, then \(\cup\tau(P(x)) \cap \sigma \neq \varnothing\).
Therefore, since \(P_{\tuple{P, \tau, Q}}\) is coarser than \(P\), then it holds \(P(x) \subseteq P_{\tuple{P, \tau, Q}}(x)\),
thus proving \(\exists E \in P \ldotp E \subseteq P_{\tuple{P, \tau, Q}}(x) \wedge (\cup\tau(E))\cap \sigma \neq \varnothing\)
and consequently~\eqref{corr-stm-expl_sound2} holds.
\end{proof}

\terminationSymbolic*
\begin{proof}
We first observe that each \emph{Search} iteration adds some new state to \(\sigma\) through the update at line~\ref{loc:sym_sigmaupdate}. Thus, since \(\Sigma\) has finitely many elements, Algorithm~\ref{algo:symbolic_sound} will always execute a finite number of \emph{Search} iterations.
We now proceed to show that the algorithm executes a finite number of \emph{Refine} iterations too.
Consider a \emph{Refine} iteration of the algorithm, the corresponding triple \(\tuple{a, B, C}\) selected for refinement, the 2PR triple \(\tuple{P, \tau, Q}\) before executing the refinement block and the 2PR triple \(\tuple{P', \tau', Q'}\) obtained right after the refinement step has been executed, then
\begin{equation}
\forall x \in B \cap \pre_a(C) \Ra R_{\tuple{P', \tau', Q'}}(x) \subset R_{\tuple{P, \tau, Q}}(x)\enspace .\label{eq:prop_termination_algosym}
\end{equation}
In fact, let \(x \in B \cap\pre_a(C) = B'\) and \(x' \in C\) be s.t. \(x \sraa{a} x'\), then
by definition of \(V\) we get \(\cup\tau(P(x)) \not\subseteq \pre_a(\cup\tau(P(x')))\).
Assume \(y \in \cup\tau(P(x)) \smallsetminus \pre_a(\cup\tau(P(x')))\) and observe that either 
\(Q(y) \cap \pre_a(\cup\tau(C)) = \varnothing\), therefore \(Q(y) = Q'(y)\) will not be split and it will be removed at line~\ref{loc:sym_refine_principal}, or \(Q(y) \cap \pre_a(\cup\tau(C)) \neq \varnothing\), therefore \(Q(y)\) will be split between lines~\ref{loc:sym_refine_Q_begin}--\ref{loc:sym_refine_Q_end} and by definition \(Q'(y) = Q(y) \smallsetminus \pre_a(\cup\tau(C))\) is such that \(Q'(y)\) will be removed at line~\ref{loc:sym_refine_principal}.
Therefore, \(Q'(y) \not\subseteq \pre_a(\cup\tau(C)) = S\), thus showing that
the execution of line~\ref{loc:sym_refine_principal} effectively refines \(\tau(B)\) by removing (some sub-block of) \(Q(y)\) from it, proving \eqref{eq:prop_termination_algosym}.
Therefore, executing the \emph{Refine} block is guaranteed to remove some state from at least one principal of the underlying relation \(R_{\tuple{P, \tau, Q}}\), and since each principal has an initial finite number of elements, the overall number of \emph{Refine} iterations is also finite.
Therefore, since every iteration of the while-loop either executes a \emph{Search} or a \emph{Refine} step, the algorithm executes a finite number of iterations.
Finally, we observe that each iteration computes a finite number of operations, so that this completes the proof.
\end{proof}

\algosymSplittingStates*
\begin{proof}
We show by induction that the following implications hold for any \(k \in \interval{0}{t}\):
\begin{align*}
\big(\forall j \in \interval{0}{k} \ldotp R_j(x) = R_j(y)\big) &\Ra P_k(x) = P_k(y) \\
\big(\forall j \in \interval{0}{t} \ldotp R_j^{-1}(x) = R_j^{-1}(y)\big) &\Ra Q_k(x) = Q_k(y)
\end{align*}
We proceed by induction on \(k\). Assume \(k = 0\), then \(R_i(x) = R_i(y)\) implies \(P_{R_i}(x) = P_{R_i}(y)\) by definition of 2PR triples and since \(P_{R_i} = P_0\) then \(P_k(x) = P_k(y)\) holds. Moreover, Property~\ref{pr:2prQO} ensures that \(P_0 = Q_0\) and thus \(Q_k(x) = Q_k(y)\) holds.
For the inductive case assume \(P_k(x) = P_k(y)\) and \(Q_k(x) = Q_k(y)\).
Assume by contradiction that \(P_{k+1}(x) \neq P_{k+1}(y)\), then this implies that the \(k+1\)-th iteration was a \emph{Refine} one in which a triple \(\tuple{a, P_k(x), C}\) was selected.
Assume \(x \in P_k(x) \cap \pre_a(C)\) and \(y \in P_k(x) \smallsetminus \pre_a(C)\) w.l.o.g., then by \eqref{eq:prop_termination_algosym}
\(R_{k+1}(x) \subseteq R_k(x)\).
Now since the \emph{Refine} block does not update the underlying principal for states in \(P_k(x) \smallsetminus \pre_a(C)\), we get \(R_k(y) = R_{k+1}(y)\) and since by hypothesis \(R_k(x) = R_k(y)\) then it follows that \(R_{k+1}(x) \subseteq R_{k+1}(y)\), and consequently \(R_{k+1}(x) \neq R_{k+1}(y)\), giving a contradiction to the hypothesis \(R_{k+1}(x) = R_{k+1}(y)\), and thus proving \(P_{k+1}(x) = P_{k+1}(y)\).
We now prove the dual result on \(Q\), assume by contradiction that \(Q_{k+1}(x) \neq Q_{k+1}(y)\), then this implies that the \(k+1\)-th iteration was a \emph{Refine} one in which a triple \(\tuple{a, B, C}\) was selected to be stabilized and \(Q_k(x) \in \tau_k(B)\).
Moreover, since \(Q_k(x)\) has been split, then both \(Q_k(x) \cap \pre_a(\cup\tau_k(C)) \neq \varnothing\) and \(Q_k(x) \not\subseteq \pre_a(\cup\tau_k(C))\) hold.
Suppose \(Q_{k+1}(x) = Q_{k}(x) \cap \pre_a(\cup\tau_k(C))\) (the case \(Q_{k+1}(y) = Q_{k}(y) \cap \pre_a(\cup\tau_k(C))\)  is symmetric), then \(Q_{k+1}(y) = Q_{k}(y) \smallsetminus \pre_a(\cup\tau_k(C))\).
Consequently, line \ref{loc:sym_refine_principal} will remove \(Q_{k+1}(y)\) from \(\tau_k(B')\), while \(Q_{k+1}(x) \in \tau_{k+1}(B')\) will still hold. 
Thus, let \(b \in B'\), then \(x \in \cup\tau_{k+1}(P_{k+1}(b))\) but \(y \not\in \cup\tau_{k+1}(P_{k+1}(b))\), which entails \(b \in R_{k+1}^{-1}(x)\) and \(b \not\in R_{k+1}^{-1}(y)\), implying \(R_{k+1}^{-1}(x) \neq R_{k+1}^{-1}(y)\), which gives a contradiction to the hypothesis \(R_{k+1}^{-1}(x) = R_{k+1}^{-1}(y)\), thus proving \(Q_{k+1}(x) = Q_{k+1}(y)\).
Finally, since Algorithm~\ref{algo:symbolic_sound} terminates after \(t\) iterations then \(P = P_t\) and \(Q = Q_t\) and the following holds:
\begin{align*}
\left(\forall j \in \interval{0}{t} \ldotp R_j(x) = R_j(y)\right) &\Ra P(x) = P(y) \\
\left(\forall j \in \interval{0}{t} \ldotp R_j^{-1}(x) = R_j^{-1}(y)\right) &\Ra Q(x) = Q(y)
\end{align*}
Finally, equations~\eqref{eq:2prsplitP}~and~\eqref{eq:2prsplitQ} follow from the above results taking the contrapositive.
\end{proof}

\section{Extension to Infinite Transition Systems}\label{sec:appendix_extension}
We now briefly discuss how to extend the proofs for Theorems~\ref{thm:explicit_sound},~\ref{thm:CorrectnessAlgoNonterminatingEx}~and~\ref{thm:symbolic_sound} to infinite state systems. 
We first note that the only points in the proofs where finiteness of the state space is relied upon are, respectively, point \textcolor{lipicsGray}{\sffamily\bfseries\upshape (iv)} of Theorem~\ref{thm:explicit_sound}, point \textcolor{lipicsGray}{\sffamily\bfseries\upshape (iii)} of Theorem~\ref{thm:CorrectnessAlgoNonterminatingEx} and point \textcolor{lipicsGray}{\sffamily\bfseries\upshape (iv)} of Theorem~\ref{thm:symbolic_sound}, therefore we proceed to show how their common pattern can be generalized to transition systems with infinitely many states (in the following we refer to the notation used in the respective proofs).
Observe that, if the input transition system is finite, then the execution of \textsf{Sim} over the output relation \(R\) is guaranteed to terminate, and correctness of \textsf{Sim} ensures that an iteration $\mathit{it}$ where some element of \(\bS\) is refined will be eventually reached in finite time, for every execution (regardless of how nondeterminism is resolved).
On the other hand, for transition systems with infinitely many states, an iteration such as $\mathit{it}$ might not exist since \textsf{Sim} might not terminate on infinite state systems and, moreover, nondeterminism might choose to refine principals only if they are not marked as reachable 
(that is, \(R(x) \cap \sigma = \varnothing\) for Theorem~\ref{thm:explicit_sound}, \(\nexists s \in \sigma \ldotp R(x) = R(s)\) for Theorem~\ref{thm:CorrectnessAlgoNonterminatingEx} and \(\cup\tau(B)\cap \sigma = \varnothing\) for Theorem~\ref{thm:symbolic_sound}), and which are, therefore, not included in \(\bS\).
However, we can use Assumption~\ref{assumption_algo4} to avoid this situation altogether since, assuming \(\bS \neq \varnothing\) (as we do, by contradiction, in the proofs), we can pick a principal \(R(z) \in \bS\) and some \(z' \in R(z) \smallsetminus \Rsim(z) \neq \varnothing\) by definition of \(\bS\).
Assumption~\ref{assumption_algo4} therefore ensures that there exists an execution where \textsf{Sim} reaches an iteration \(\mathit{it_1}\) where \(z'\) is removed from the principal of \(z\).
Then, let \(R_1\) be the current relation at iteration \(\mathit{it_1} + 1\) and define \(\bS' = \{R(x) \mid x \in \Sigma,\, R(x) \in \bS,\, R_1(x) \neq R(x)\}\) to be the set of principals which have been refined by \textsf{Sim} up to iteration \(\mathit{it_1} + 1\).
We find that \(\bS'\) is nonempty by construction since \(R(z) \in \bS'\), and that \(\Rsim \subseteq R_1\subseteq R\) following Inv\textsubscript{1}.
The proof then proceeds as explained in the respective points \textcolor{lipicsGray}{\sffamily\bfseries\upshape (iv)}, \textcolor{lipicsGray}{\sffamily\bfseries\upshape (iii)} and \textcolor{lipicsGray}{\sffamily\bfseries\upshape (iv)} by considering the first principal in \(\bS'\) to be refined and the corresponding iteration \(\mathit{it}\), allowing us to get a contradiction which proves that \(\bS = \varnothing\).%

\section{Experimental Evidence for 2PR Triples}\label{sec:experiments}
We implemented the algorithms presented in this paper for finite transition systems that are explicitly represented.
In this section, by leveraging this implementation, we empirically compare Algorithm~\ref{algo:explicit_sound}, which uses an explicit representation for \(R\), and Algorithm~\ref{algo:symbolic_sound}, which symbolically represents \(R\) through 2PR triples. 
For the benchmarks we used a small subset of the BEEM database for explicit model checkers~\cite{Pelnek}.
The selected subset of benchmarks consists of mutual exclusion protocols that we modified to increase the number of unreachable states by performing an ``unrolling'' of the protocol loop and then by picking an initial state after the unrolled protocols.
The rationale behind this modification is that, typically, benchmarks available online as explicit transition systems have no unreachable states and, often, are strongly connected, which means that moving the initial state has no effect on the number of unreachable states.

Our experimental results are summarized in Tables~\ref{tab:results}~and~\ref{tab:results2} where the initial equivalence/preorder is set to \(R_i = \Sigma \times \Sigma\), and \(\sigma_i=\varnothing\).
Following their definitions, we infer that the following inequalities must hold for the entries of Tables~\ref{tab:results}~and~\ref{tab:results2} 
\[
  \Sigma \geq P^1 \geq r^1 \enspace \text{, and}\enspace\Sigma \geq P^2 \geq P_{P, \tau, Q}^2 \geq r^2 \enspace .
\]

\begin{table*}[t]
\caption{%
  Execution metrics relative to Algorithm~\ref{algo:explicit_sound}.
  The column \(\mathord{\rightarrow}\) is the number of labeled transitions; 
  \(\Sigma\) the size of the state space;
  \(\sigma^1\) the size of \(\sigma\) as computed by Algorithm~\ref{algo:explicit_sound}, 
  \(P^1\) the size of \(P\) as computed by Algorithm~\ref{algo:explicit_sound} as defined in Theorem~\ref{thm:explicit_sound}, and \(t^1\) is the execution time in seconds;
  \(r^1\) is the number of blocks in the output of Algorithm~\ref{algo:explicit_sound} as per the right hand side of \eqref{corr-stm-expl_sound2}: \(|\{B \in P \mid R(B)\cap \sigma \neq \varnothing\}|\);
  The symbol \(\timeout\) indicates execution timeout after 24 hours.
}
\label{tab:results}
\centering
\begin{tabular}{|l|c|c|c|c|c|c|}\hline
\rule{0pt}{2.5ex}    
  \textbf{Protocol} & ~~~$\rightarrow$~~~ &  ~~~$\Sigma$~~~ & $\sigma^1$ & $P^1$ & $r^1$ & $t^1$ \\[3.5pt] \hline\hline
  \begin{filecontents*}{stats.csv}
name, states, edges, esigma, ep, ereach, etime, ssigma, sp, sq, spptq, sreach, stime, gain, psimr
bakery.1, 1628, 2917, 60, 464, 396, 246.73, 59, 672, 486, 466, 396, 28.66, 88.4, 142
bakery.2, 1340, 2453, 96, 769, 638, 145.28, 99, 779, 752, 741, 638, 72.51, 50.1, 504
fischer.1, 1254, 2748, 86, 359, 295, 117.17, 91, 675, 512, 511, 295, 62.21, 46.9, 266
fischer.2, 42272, 131407, \(\timeout\), \(\timeout\), \(\timeout\), \(\timeout\), \(\timeout\), \(\timeout\), \(\timeout\), \(\timeout\), \(\timeout\), \(\timeout\), \_, 14481
mcs.1, 11839, 31645, 165, 1185, 447, 59333.74, 172, 766, 793, 661, 447, 253.93, 99.6, 307
mcs.2, 208, 480, 39, 39, 39, 0.72, 37, 42, 39, 39, 39, 0.07, 89.9, 39
peterson.1, 17919, 47064, \(\timeout\), \(\timeout\), \(\timeout\), \(\timeout\), 928, 2012, 2176, 1787, 1280, 4211.78, \(\infty\), 1275
\end{filecontents*}
\csvreader[
  late after line = \\\hline
  ]{stats.csv}{
  name=\name, edges=\edges, states=\states, esigma=\esigma, ep=\ep, ereach=\ereach, etime=\etime}{%
  \name & \edges & \states & \esigma & \ep &\ereach & \etime 
}%
\end{tabular}
\end{table*}

\begin{table*}[t]
\caption{%
  Execution metrics relative to Algorithm~\ref{algo:symbolic_sound}, and comparison with those of Algorithm~\ref{algo:explicit_sound}.
  The column \(\mathord{\rightarrow}\) is the number of labeled transitions; 
  \(\Sigma\) the size of the state space;
  \(\sigma^2\), \(P^2\), \(Q^2\), \(P_{\tuple{P, \tau, Q}}^2\) are the sizes of \(\sigma\), \(P\), \(Q\) and \(P_{\tuple{P, \tau, Q}}\), respectively, as computed by Algorithm~\ref{algo:symbolic_sound}, and \(t^2\) is the execution time in seconds;
  \(r^2\) is the number of blocks in the output of Algorithm~\ref{algo:symbolic_sound} as per the right hand side of \eqref{corr-stm-sym_sound2}: \(|\{B \in P_{\tuple{P, \tau, Q}} \mid \exists E \in P. \: E \subseteq B \wedge (\cup \tau(E))\cap \sigma \neq \varnothing\}|\);
  {\rm \textbf{gain}} is the percentage decrease in execution time of Algorithm~\ref{algo:symbolic_sound} w.r.t. Algorithm~\ref{algo:explicit_sound}, i.e. \((t^1 - t^2)/t^1\) where \(t^1\) is the corresponding metric of Table~\ref{tab:results}.
  The symbol \(\timeout\) indicates execution timeout after 24 hours.
}
\label{tab:results2}
\centering
\begin{tabular}{|l|c|c|c|c|c|c|c|c|c|}\hline
\rule{0pt}{2.5ex}    
  \textbf{Protocol} & ~~~$\rightarrow$~~~ &  ~~~$\Sigma$~~~ & $\sigma^2$ & $P^2$ & $Q^2$ & $P_{\tuple{P, \tau, Q}}^2$ & $r^2$ & $t^2$ & \textbf{gain} \%\\[3.5pt] \hline\hline
  \csvreader[
  late after line = \\\hline
  ]{stats.csv}{
  name=\name, edges=\edges, states=\states, ssigma=\ssigma, sp=\sp, sq=\sq, spptq=\spptq, sreach=\sreach, stime=\stime, gain=\gain}{%
  \name & \edges & \states & \ssigma & \sp & \sq & \spptq & \sreach & \stime & \gain
}%
\end{tabular}
\end{table*}

The results in Tables~\ref{tab:results}~and~\ref{tab:results2} show that Algorithm~\ref{algo:symbolic_sound} has a significant gain in term of execution time (reduction ranges between a minimum of \(46.9\%\) and a maximum of \(99.6\%\) less time required for computation).
Of course, it is worth pointing out that the set of benchmarks is fairly small and that every transition system is a mutual exclusion protocol that has been unrolled. 

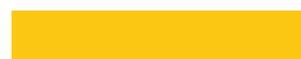
\end{document}